\documentclass[reqno]{amsart}

\def\eps{\epsilon}%
\def\tensor{\,\raise2pt\hbox{${}_{\otimes}$}\,}
\def\fdg{\,:\,}
\def\ptl{\partial}
\def\rest#1{\raise-2pt\hbox{${\lfloor_{#1}}$}}

\def\mbo#1{\boldsymbol{#1}}

\def\ip#1#2{\langle#1,#2\rangle}

\def\olin#1{\overline{#1}{}}

\def\grad{{\nabla}}
\newcommand{\leftexp}[2]{{\vphantom{#2}}^{#1}{#2}}

\def\halb{\frac{1}{2}}

\def \a{\alpha}
\def \b {\beta}

\newtheorem{theorem}{Theorem}[section]
\newtheorem{lemma}[theorem]{Lemma}

\newtheorem{corollary}[theorem]{Corollary}
\newtheorem{remark}[theorem]{Remark}
\newtheorem{definition}[theorem]{Definition}
\newtheorem{claim}[theorem]{Claim}

\newcommand{\ba}{\begin{array}}
\newcommand{\ea}{\end{array}}
\newcommand{\bea}{\begin{eqnarray}}
\newcommand{\eea}{\end{eqnarray}}
\newcommand{\bee}{\begin{eqnarray*}}
\newcommand{\eee}{\end{eqnarray*}}

\renewcommand{\vec}[1]{\mbox{\boldmath $#1$}}

\renewcommand{\a}{\alpha}
\renewcommand{\b}{\beta}

\renewcommand{\r}{\rho}

\usepackage{color}

\newcommand{\green}[1]{{\color{green}#1}}

\newcounter{mnotecount}[section]

\renewcommand{\themnotecount}{\thesection.\arabic{mnotecount}}

\newcounter{mymnotecount}[section]
\renewcommand{\themymnotecount}{\thesection.\arabic{mymnotecount}}
\newcommand{\mymnote}[1]{\protect{\stepcounter{mymnotecount}}${\raisebox{0.5\baselineskip}[0pt]{\makebox[0pt][c]{\color{green}{\tiny\em$\bullet$\themnotecount}}}}$\marginpar{\raggedright\tiny\em$\!\bullet$\themymnotecount:

\green{#1}}\ignorespaces}

\renewcommand{\mymnote}[1]{}
\usepackage{mathrsfs}
\begin{document}

\title[ ADM Energy for Axially Symmetric Perturbations of Kerr Metric]{ Axially Symmetric Perturbations of Kerr Black Holes I: A gauge-invariant construction of ADM Energy} 


\author{Nishanth Gudapati}
\address{Center of Mathematical Sciences and Applications, Harvard University, 20 Garden Street, Cambridge, MA-02138, USA}
\email{nishanth.gudapati@cmsa.fas.harvard.edu}



\begin{abstract}
Based on the Hamiltonian dimensional reduction of $3+1$ axially symmetric, Ricci-flat Lorentzian spacetimes to a $2+1$ Einstein-wave map system with the (negatively curved) hyperbolic 2-plane target, we construct a positive-definite, (spacetime) gauge-invariant energy functional for linear axially symmetric perturbations in the exterior of Kerr black holes, in a manner that is also gauge-independent on the target manifold. We also show that the positive-definite energy functional serves as a Hamiltonian for the constrained evolution of the linear perturbations. 

\end{abstract}

\maketitle

\section{Geometric Mass-Energy and Perturbations of Black Holes}
The stability of stationary solutions of a physical law serves as an impetus to the validity of the law. In the context of Einstein's equations for general relativity, an important stationary solution is the  Kerr family of black holes 
which is also an asymptotically flat, axially symmetric family of solutions of the $3+1$ dimensional vacuum Einstein equations for general relativity: 

\begin{align}\label{EVE}
\bar{R}_{\mu \nu} =0, \quad (\bar{M}, \bar{g}). 
\end{align}
In parts due to the physical relevance and the mathematical beauty arising from its multiple miraculous properties (see e.g., \cite{Teukolsky_15, Chandrasekhar_83}), the problem of stability of Kerr black hole spacetimes within the class of  Einstein's equations \eqref{EVE} has been a subject of active research interest since their discovery by R. Kerr in 1963. However, geometric properties of Kerr black hole spacetimes such as stationarity (as opposed to staticity), trapping of null geodesics and the general issue of gauge dependence of metric perturbations cause significant obstacles in the resolution of this `black hole stability' problem. 
In this work, we focus on the issue caused by the stationarity of the Kerr metric: 

\begin{description}
\item [\textbf{(P1)}] The problem of the ergo-region, the lack of a positive-definite and conserved energy and the superradiance, caused by the shift vector of the Kerr metric. 
\end{description}

It may be noted that  an asymptotically flat spacelike Riemannian hypersurface $(\olin{\Sigma}, \bar{q})$ such that $\bar{M} = \olin{\Sigma} \times \mathbb{R}$ satisfies the Einstein's equations for general relativity \eqref{EVE}, has a positive-definite total (ADM) mass $m_{\text{ADM}}:$

\begin{align}\label{ADM-def}
m_{\text{ADM}} \fdg =  \lim_{r \to \infty} \int_{\mathbb{S}^2(r)} \sum^3_{i,j,k=1} (\ptl_k \bar{q}_{i \ell} - \ptl_i \bar{q}_{\ell k}) \frac{x^i}{r} \bar{\mu}_{\mathbb{S}^2}, \quad \bar{q} \quad \text{is asymptotically Euclidean}
\end{align}

 from the celebrated positive-mass theorems of Schoen-Yau and Witten \cite{schoen-yau-1,schoen-yau-2,witten-pmt}. However, it is not necessary that the positivity of energy carries forward to the perturbative theory of Einstein's equations. In a general asymptotically flat manifold, it is a priori not determinate whether the mass-energy at infinity increases or decreases for small perturbations.  
 This outcome can be seen in the energy of a (linear) scalar wave equation propagating in the exterior of Kerr black holes - an illustrative, albeit a special `test' case of the perturbations of the Kerr metric. 
 
 However, as we already alluded to, the difficulty of constructing a positive-definite energy for the perturbative theory is not only due to the shift vector (or the ergo-region) of the Kerr metric. Even if one considers the Schwarzschild metric (the special case of vanishing angular momentum of Kerr),
 
 \begin{align}\label{sch}
 \bar{g} =- f^{-1} dt^2+ f dr^2 + r^2 d\omega^2_{\mathbb{S}^2},
 \end{align}
 where $f \fdg = (1-2mr^{-1}),$ it is not immediate that there exists a positive-definite energy for the perturbative theory of \eqref{sch}. In 1974, Moncrief had devised a `Hamiltonian' for the perturbative theory of Schwarzschild based on the ADM formalism of Einstein's equations \cite{Moncrief_74}. Suppose the Lorentzian spacetime $(\bar{M}, \bar{g})$ admits a $3+1$ ADM decomposition: 
\begin{align}
\bar{g} = - N^2 dt^2 + \bar{q}_{ij} (dx^i + N^i dt) \otimes (dx^j + N^j dt)
\end{align}
then the ADM constraint and evolution equations are given by the variational principle for the phase space $X_{\text{ADM}} \fdg = \{ (\mbo{\bar{\pi}}^{ij}, \bar{q}_{ij}),\, i, j= 1, 2, 3 \}$: 

\begin{subequations}
\begin{align}
I_{\text{ADM}} \fdg =& \int \left( \bar{\mbo{\pi}}^{ij} \ptl_t \bar{q}_{ij} - N H -N^i H_i \right)d^4x, 
\intertext{where}
H \fdg=& \bar{\mu}^{-1}_{\bar{q}} \left( \Vert \mbo{\bar{\pi}} \Vert^2_{\bar{q}} - \halb \text{Tr}_{\bar{q}} (\mbo{\bar{\pi}})^2 \right) - \bar{\mu}_{\bar{q}} R_{\bar{q}} \\
H_i \fdg=& -2 \leftexp{(\bar{q})}{\grad}_j \mbo{\bar{\pi}}^j_i
\end{align}
\end{subequations}
and $\{ N, N^i \}$ are the Lagrange multipliers. Suppose $(\bar{q}, \bar{\mbo{\pi}})$ are such that $\bar{q} -\bar{\delta} \in H^2_{-1/2}$ and  $ \bar{\mbo{\pi}} \in H^1_{-3/2}$  asymptotically flat i.e., diffeomorphic to $\mathbb{R}^3 \setminus B_1(0)$ in the complement of a compact set in $\olin{\Sigma}$ and
\begin{align}
q_{ij} =& \left(1+ \frac{M}{r} \right)  \bar{\delta}_{ij} + \mathcal{O}(r^{-1-\a}) \notag\\
\bar{\mbo{\pi}}^{ij} =& \mathcal{O}(r^{-2-\a}),
\end{align}
$ r= \vert x \vert,$ for some $\a>0,$ outside the compact set in $\olin{\Sigma}.$  As a consequence $M =m_{\text{ADM}}.$

It follows that the evolution equations are given by
\begin{align}
\ptl_t \bar{q}_{ij} =& 2 \bar{N} \bar{\mu}_{\bar{q}} \left(\bar{\mbo{\pi}}_{ij} - \halb q_{ij} \text{Tr}_{\bar{q}}(\bar{\mbo{\pi}}) \right) + \leftexp{(\bar{q})}{\grad}_j \bar{N}_i+ \leftexp{(\bar{q})}{\grad}_i \bar{N}_j \\
\ptl_t \bar{\mbo{\pi}}^{ij} =& - \bar{N} \bar{\mu}_{\bar{q}} \left( \leftexp{(\bar{q})}{R}^{ij} - \halb \bar{q}^{ij} R_{\bar{q}}  \right ) + \halb \bar{N} \bar{\mu}^{-1}_{\bar{q}} q^{ij} \left( \text{Tr}_{\bar{q}} (\bar{\mbo{\pi}}^2) - \halb \text{Tr}_{\bar{q}}(\bar{\mbo{\pi}})^2 \right) 
\notag\\
&-2 \bar{N} \bar{\mu}^{-1}_{\bar{q}} (\bar{\pi}^{im} \bar{\mbo{\pi}}^j_m - \halb \bar{\mbo{\pi}}^{ij} \text{Tr}_{\bar{q}}(\bar{\mbo{\pi}})) + \bar{\mu}_q ( \leftexp{(\bar{q})}{\grad}^
i \leftexp{(\bar{q})}{\grad}^
j \bar{N} \notag\\
&- \bar{q}^{ij} \leftexp{(\bar{q})}{\grad}
^m \leftexp{(\bar{q})}{\grad}_m \bar{N}) + \leftexp{(\bar{q})}{\grad}_m (\bar{\mbo{\pi}}^{ij} \bar{N} ^m) - \leftexp{(\bar{q})}{\grad}_m \bar{N}^i \bar{\mbo{\pi}}^{mj} 
- \leftexp{(\bar{q})}{\grad}_m \bar{N}^j \bar{\mbo{\pi}}^{mi}
\end{align}
Suppose we consider the small perturbations of the initial data of Schwarzschild black hole spacetimes: $\bar{q}= \bar{q}_{s} + \eps q'$ and $\bar{\mbo{\pi}}= \bar{\mbo{\pi}}_s + \eps \bar{\mbo{\pi}}',$ Moncrief's Hamiltonian energy formula is 

\begin{align}\label{vm-74}
H_{\text{pert}} \fdg =& \int_{\olin{\Sigma}} \Big \{ \bar{N} \bar{\mu}^{-1}_{\bar{q}} \Big(\Vert \bar{\mbo{\pi}}' \Vert_{\bar{q}}^2 - \halb \text{Tr}_{\bar{q}}(\bar{\mbo{\pi}}')^2 \Big) + \halb \bar{N} \bar{\mu}_{\bar{q}} \Big(\halb \leftexp{(\bar{q})}{\grad}_k \bar{q}'_{ij} \leftexp{(\bar{q})}{\grad}^k \bar{q}'^{ij} \notag\\
& \quad - \leftexp{(\bar{q})}{\grad}_k \bar{q}'_{ij} \leftexp{(\bar{q})}{\grad}^j \bar{q}'^{ik} 
- \halb \leftexp{(\bar{q})}{\grad}_i \bar{q}' \leftexp{(\bar{q})}{\grad}^i \bar{q}'
+ 2\leftexp{(\bar{q})}{\grad}_i \bar{q}' \leftexp{(\bar{q})}{\grad}_j \bar{q}'^{ij} \notag\\
& \quad + \bar{q}' \leftexp{(\bar{q})}{\grad}^2_{ij} \bar{q}'^{ij} - \bar{q}' \bar{q}'_{ij} R^{ij}_{\bar{q}} \Big) \Big \} d^3x,
\end{align}
which is a volume integral on the hypersurface $\olin{\Sigma},$ where $\bar{q}' = \text{Tr}(\bar{q}'_{ij}).$
Moncrief used the Hamiltonian formulation to decompose the metric perturbations into gauge-dependent, gauge-independent and constraints; and ultimately reconciled with the Regge-Wheeler-Zerilli results \cite{Regge-Wheeler_57, Zerilli_70}. An important feature of these results  is that  the energy functional \eqref{vm-74} can be realized to be positive-definite for both odd and even parity perturbations. Using tensor harmonics, positive-definite energy functionals for both odd and even parity  perturbations of Schwarzschild black holes were constructed in \cite{Moncrief_74}.  In this spirit, a number of pioneering articles on the perturbations of static black holes were written by Moncrief \cite{Moncrief_74_1, Moncrief_74_2, Moncrief_74_3}.

The subject of this article is to focus on axially symmetric perturbations of the Kerr metric. In precise terms, the Kerr metric $(\bar{M}, \bar{g})$ can be represented in Boyer-Lindquist coordinates $(t, r, \theta, \phi)$ as  

 \begin{align}\label{BL-Kerr}
\bar{g} =& - \left( \frac{\Delta - a^2 \sin^2 \theta}{\Sigma} \right) dt^2 - \frac{2a \sin^2 \theta (r^2 + a^2 -\Delta)}{\Sigma} dt d\phi \notag\\
&+ \left( \frac{(r^2 +a^2)^2 -\Delta a^2 \sin^2 \theta}{\Sigma}\right) \sin^2 \theta d\phi^2 + \frac{\Sigma}{\Delta} dr^2  + \Sigma d\theta^2
\end{align}
where,
\begin{subequations}
\begin{align}
\Sigma \fdg =&\, r^2 + a^2 \cos^2 \theta \\
\Delta \fdg =&\, r^2 -2Mr + a^2, \quad \textnormal{with the real roots} \quad \{r_-,r_+\} \\
r_{+} \fdg =&\, M + \sqrt{M^2- a^2} > r_{-} \notag
\intertext{and}
\theta \in [0, \pi],&\quad r \in (r_+, \infty),\quad \phi \in [0, 2\pi).
\end{align}
\end{subequations}
It is well known that Einstein's equations  \eqref{EVE} on spacetimes $(\bar{M}, \bar{g})$ with one isometry $(\frac{\ptl}{\ptl \phi})$, represented in Weyl-Papapetrou coordinates, 
\begin{align}
\bar{g} = e^{-2\gamma} g + e^{2\gamma} (d \phi + A_\nu dx^\nu)^2,
\end{align}
admit a dimensional reduction to a 2+1 dimensional Einstein wave map system 

\begin{subequations} \label{ewm-system}
\begin{align}
E_{\mu \nu} =& \, T_{\mu \nu},  \\
\square_g U^A + \leftexp{(h)}{\Gamma}^A_{BC} g^{\mu \nu} \ptl_\mu U^B \ptl_\nu U^C=& \,0,  \quad \text{on} \quad (M, g).
\end{align}
\end{subequations}
where $\square_g$ is the covariant wave operator, $E_{\mu \nu}$ the Einstein tensor in the interior of the quotient $(M, g) \fdg= (\bar{M}, \bar{g})/ SO(2)$ and $T$ is the stress energy tensor of the wave map $U \fdg (M, g) \to (\mathbb{N}, h),$ $\mathbb{N}$ is the negatively curved hyperbolic $2$-plane, 

\begin{align}
T_{\mu \nu} = \ip{\ptl_\mu U} {\ptl_\nu U }_{h(U)} - \halb g_{\mu \nu} \ip{\ptl_\sigma U}{ \ptl^\sigma U}_{h(U)}.
\end{align}
Introducing the coordinates $(\r, z)$ such that $\r = R \sin \theta$  $z= R \cos 
\theta,$ where $R \fdg= \halb (r-m+ \sqrt{\Delta}),$ the Kerr metric \eqref{BL-Kerr} can be represented in the Weyl-Papapetrou form as 

\begin{align}
\bar{g} = \Sigma \zeta^{-1} (-\Delta dt^2 + \zeta R^{-2} (d\r^2 + dz^2)) + \sin^2\theta \Sigma^{-1} \zeta
(d \phi - 2aMr \zeta^{-1} dt))^2
\end{align}
where $\zeta = (r^2 +a^2)^2 - a^2 \Delta \sin^2 \theta.$ Furthermore, the Kerr metric can also be represented in the Weyl-Papapetrou form using functions $(\bar{\r}, \bar{z})$ 
such that 
\begin{align}
\bar{\r} = \r - \frac{(m^2 -a^2)}{4 R^2}\r, \quad \text{and} \quad \bar{z} = z + \frac{(m^2 -a^2)}{4 R^2}z
\end{align}
(cf. Appendix A in \cite{GM17} for details). Now we shall turn to the axially symmetric perturbation theory of the Kerr metric. 

In view of the peculiar behaviour of the 2+1 Einstein-wave map system, a detailed discussion of our methods is relevant for our article and perhaps also interesting to the reader. Consider the Hamiltonian energy of an axially symmetric linear wave equation propagating on the Kerr metric ($\square_g u =0$),
\begin{align}
H^{\text{LW}} \fdg = \int_{\olin{\Sigma}}  \left( \halb \bar{N}\bar{\mu}^{-1}_q v^2 + vN^i \ptl_i u+  \halb N \bar{\mu}_q \bar{q}^{ij} \ptl_i u \ptl_j u  \right) d^3 x
\end{align}
where $v$ is the conjugate momentum of $u,$ the energy is directly positive-definite. 
However this simplification does not carry forward to the Maxwell equations on the Kerr metric

\begin{align}
H^{\text{Max}} \fdg =&  \int \left( \halb N \bar{q}_{ij} \bar{\mu}_q (\mathfrak{E}^i \mathfrak{E}^j + \mathfrak{B}^i \mathfrak{B}^j) + N^i \eps_{ijk}\mathfrak{E}^j \mathfrak{B}^k   \right) d^3 x
\intertext{where}
\mathfrak{E}^i \fdg =& \halb \eps^{ijk} \leftexp{*}F_{jk}, \quad
\mathfrak{B}^i \fdg = \halb \eps^{ijk} F_{jk}.
\end{align}
Actually, one can construct counter examples of positivity of energy density, for instance using, time-symmetric Maxwell fields (cf. the discussion in Section 2 in\cite{GM17}).  In a crucial work, Dain-de Austria \cite{DA_14} had arrived at a positive-definite energy for the gravitational perturbations of extremal Kerr black holes using the Brill mass formula \cite{D09} and subsequent use of Carter's identity, originally developed for black hole uniqueness theorems. In a Weyl coordinate system for the spacelike hypersurface $(\olin{\Sigma}, q)$ in extremal Kerr spacetime, their positive-definite energy for axially symmetric perturbations is obtained
from perturbations of the Brill mass formula, which in turn is obtained from multiplying a factor with the Hamiltonian constraint that conveniently results in a volume form (in the chosen Weyl coordinate system) that is useful in its representation.

In order to construct a positive-definite energy for the perturbations of Kerr-Newman metric for the full-subextremal range, we  delve into the variational structure of the relevant field equations. The beautiful linearization stability framework, developed by V. Moncrief, J. Marsden and A. Fischer \cite{FM_75,Mon_75, FMM_80}, provides a 
natural mechanism to construct an energy-functional based on the kernel of the adjoint of the deformations around the Kerr metric of the dimensionally reduced constraint map. This recognition allows us to extend results to the full sub-extremal range $(\vert a \vert, \vert Q \vert <M)$ of the perturbations of the Kerr-Newman metric \cite{GM17}, which is a solution of Einstein-Maxwell equations of general relativity. 



 Consider the ADM decomposition of $\bar{M}=\olin{\Sigma} \times \mathbb{R}.$ Suppose the group $SO(2)$ acts on $\olin{\Sigma}$ through isometries such that $\Gamma$ is the fixed point set. Suppose the norm squared of the Killing vector generating the rotational isometry is denoted by $e^{2\gamma}$. Let the Lorentzian manifold with boundary $\Sigma \times \mathbb{R}$ be denoted as $(M, g).$   In the dimensional reduction ansatz, the metric $\bar{g}$ is
\begin{align}\label{KK-ADM}
\bar{g} = e^{-2\gamma} (-N^2 dt^2 + q_{ab} (dx^a + N^a dt) \otimes (dx^b + N^b dt)) + e^{2\gamma} (d\phi + \mathcal{A}_0 dt + \mathcal{A}_a dx^a)^2.
\end{align}
In the dimensional reduction framework, identifying the reduced conjugate momenta, which form the reduced phase space in $(M, g);$ and the corresponding reduced Hamiltonian formalism is nontrivial. This construction was  done in \cite{kaluz1}. Define the conjugate momentum corresponding to the metric $q_{ab}$ as follows: 
\begin{align}
\mbo{\pi}^{ab} = e^{-2\gamma} \bar{\mbo{\pi}}^{ab}, \quad \bar{q}_{ab} = e^{-2\gamma} q_{ab} + e^{2\gamma} \mathcal{A}_a \mathcal{A}_b.
\end{align} 
As a consequence, the ADM action principle transforms to 
\begin{align}
J = \int^{t_2}_{t_1}\int_{\Sigma} \big(   \mbo{\pi}^{ab} \ptl_t q_{ab} + \mathcal{E}^a \ptl_t \mathcal{A}_a + p \ptl_t \gamma - N H - N^aH_a+ \mathcal{A}_0 \ptl_a \mathcal{E}^a \big) d^2xdt
\end{align}
where the phase-space is now 
\begin{align}
\big \{ (q, \mbo{\pi}), (\mathcal{A}_a, \mathcal{E}^a), (\gamma, p) \big\}
\intertext{with the Lagrange multipliers}
\big \{ N, N^a, \mathcal{A}_0 \big\}
\end{align}
and the constraints: 
\begin{subequations} \label{2+1constraints}
\begin{align}
H=& \bar{\mu}^{-1}_q (\Vert \mbo{\pi} \Vert^2_q - \text{Tr}_q(\mbo{\pi})^2) + \frac{1}{8} p^2 + \halb e^{-4\gamma}
q_{ab} \mathcal{E}^a \mathcal{E}^b + \bar{\mu}_q (-R_q + 2 q^{ab} \ptl_a \gamma \ptl_b \gamma) \notag\\
&+ \frac{1}{4} e^{4\gamma} q^{ab}q^{bd} \ptl_{[b} \mathcal{A}_{a]}  \ptl_{[d} \mathcal{A}_{c]}, \\
H_a=& -2 \leftexp{(q)}{ \grad}_b \mbo{\pi}^b_a + p \ptl_a \gamma + \mathcal{E}^b (\ptl_{[a}\mathcal{A}_{b]} ) ,\\
\ptl_a \mathcal{E}^a =&0
\end{align}
\end{subequations}
After applying the Poincar\`e Lemma and introducing the twist potential such that $\mathcal{E}^a = \fdg \eps^{ab} \ptl_b \omega$ we transform into the phase space 

\[ X_{\text{EWM}} = \big \{ (\gamma, p), (\omega, \mbo{r}), (q_{ab}, \mbo{\pi}^{ab}) \big \}\] and the variational principle reduces to 

\begin{subequations} \label{2+1-no-constraints}
\begin{align}
\tilde{J} \fdg=&  \int^{t_2}_{t_1}\int_{\Sigma} \left(  \mbo{\pi}^{ab} \ptl_t q_{ab} + p \ptl_t \gamma + r \ptl_t \omega - N H - N^a H_a \right) d^2x dt,
\intertext{where $H$ and $H_a$ are now}
H=&  \bar{\mu}^{-1}_q \left(  \Vert \mbo{\pi} \Vert^2_q - \text{Tr}_q (\mbo{\pi})^2 + \frac{1}{8} p^2 + \frac{1}{2} e^{4\gamma} \mbo{r}^2 \right) \notag\\
& \quad + \bar{\mu}^{-1}_q \Big( -R_q + 2 q^{ab} \ptl_a \gamma \ptl_b \gamma + \halb e^{-4\gamma} q^{ab} \ptl_a \omega \ptl_b \omega \Big) \\
H_a=& -2 \leftexp{(q)}{\grad}_b \mbo{\pi}^b_a + p \ptl_a \gamma + \mbo{r} \ptl_a \omega.
\end{align}
\end{subequations}
\noindent with the Lagrange multipliers $N, N_a.$ After computing the field  equations in the perturbed phase space \[ X' \fdg = \big \{(\gamma', p'), (\omega', \mbo{r}'), 
(q'_{ab}, \mbo{\pi}^{' ab} ) \big\} \] 
It was noted that $ (N, 0)^{\text{T}}$ is an element of the kernel of the adjoint of the perturbed constraint map. This in turn provides a candidate for the energy, analogous to \eqref{vm-74}. The resulting expression has the potential energy
\begin{align}
D^2 \cdot \mathcal{V} = \bar{\mu}_q q^{ab} \big( 4 \ptl_a \gamma'
\ptl_b \gamma'  + e^{-4\gamma} \ptl_a \omega' \ptl_b \omega'   + 8e^{-4\gamma} \gamma'^2 \ptl_a \omega \ptl_b \omega-8e^{-4\gamma} \gamma' \ptl_a \omega \ptl_b \omega'\big).
\end{align}
which is then transformed to a positive-definite form using the Carter-Robinson identities. Firstly, it may be noted that, in the original Carter-Robinson identities are not restrictive to the choice of the function `$\r$' (   in \cite{Car_71} and in eq. (5) in \cite{Rob_74} ) and can thus be generalized as follows 

\begin{align}
&\bar{\mu}_q q^{ab} \big( 4 \ptl_a \gamma'
\ptl_b \gamma'  + e^{-4\gamma} \ptl_a \omega' \ptl_b \omega'   + 8e^{-4\gamma} \gamma'^2 \ptl_a \omega \ptl_b \omega-8e^{-4\gamma} \gamma' \ptl_a \omega \ptl_b \omega'\big) \notag\\
&+ \ptl_b (N \bar{\mu}_q q^{ab} ( -2 e^{-4\gamma} \ptl_a \gamma \omega' + e^{-4\gamma} \omega' + 4 e^{-4\gamma} \gamma' \ptl_a \omega)) \notag\\
&+ \halb \bar{\mu}_q e^{-4\gamma} L_1 (e^{-2\gamma} \omega') + \bar{\mu}_q L_2 (-4\gamma' \omega')  \notag\\
&= N \bar{\mu}_q q^{ab} (\leftexp{(1)}{V}_a \leftexp{(1)}{V}_b + \leftexp{(2)}{V}_a \leftexp{(2)}{V}_b + \leftexp{(3)}{V}_a \leftexp{(3)}{V}_b),
\end{align}
where,
\begin{subequations}
\begin{align}
\leftexp{(1)}{V}_a =& 2 \ptl_a \gamma' + e^{-4\gamma} \omega' \ptl_a \omega, \\
\leftexp{(2)}{V}_a=& -\ptl_a (e^{-2\gamma}\omega') + 2 e^{-2\gamma}\gamma' \ptl_a \omega, \\
\leftexp{(3)}{V}_a =& 2 \ptl_a \gamma \omega' -2\gamma' \ptl_a \omega,
\end{align}
\end{subequations}
and 
\begin{subequations}
\begin{align}
L_1 \fdg =& e^{-2 \gamma} ( \ptl_b (N \bar{\mu}_q q^{ab} \ptl_a \gamma) + N e^{-4\gamma} \bar{\mu}_q q^{ab} \ptl_a \omega \ptl_b \omega) \\
L_2 \fdg =& - \ptl_b( N \bar{\mu}_q q^{ab} e^{-4\gamma} \ptl_a \omega) 
\end{align}
\end{subequations}
results in a positive-definite energy of the form
\begin{align}
H^{\text{Reg}} = \int_{\Sigma} &  \Big \{ N \bar{\mu}^{-1}_q   \Big(
\varrho'^{b}_a \varrho'^{a}_b + \frac{1}{8} p'^2+ \halb
e^{4\gamma} \mbo{r}'^2  \Big) - \halb \bar{\mu}_q \tau'^2 \Big) \notag\\
&+ N \bar{\mu}_q q^{ab} \Big( 2 (\ptl_a \gamma' + \halb e^{-4\gamma} \omega' \ptl_a \omega) (\ptl_b \gamma' + \halb e^{-4\gamma} \omega' \ptl_b \omega) \notag\\
&+ 2(\gamma' e^{-2\gamma}\ptl_a \omega - \ptl_a (e^{-2\gamma} \omega'))(\gamma' e^{-2\gamma}\ptl_b \omega - \ptl_b (e^{-2\gamma} \omega' )) \notag\\ 
 &+ 2e^{-4\gamma} ( \ptl_a \gamma \omega' - \gamma' \ptl_a \omega)( \ptl_b \gamma \omega' - \gamma' \ptl_b \omega) \Big)
\Big\}\, d^2 x
\end{align}
modulo a time-coordinate gauge condition $\tau'=0.$
It is then shown that this energy functional is a Hamiltonian for the dynamics of the reduced Einstein equations in the perturbative phase-space and a spacetime divergence-free vector field density is constructed:

\begin{align}
J^{\text{Reg}} = (J^{\text{Reg}})^t \ptl_t + (J^{\text{Reg}})^a \ptl_a 
\end{align}
where $(J^{\text{Reg}})^t = \mathbf{e}^{\text{Reg}}$ and 

\begin{align}
(J^{\text{Reg}})^a =& N^2 \bar{\mu}^{-1}_q \big ( (p' \bar{\mu}_q q^{ab} \ptl_b \gamma') + e^{4\gamma} \mbo{r}' (e^{-4\gamma} \bar{\mu}_q q^{ab} \ptl_b \omega' ) \big) + \gamma' \mathcal{L}_{N'} (4N \bar{\mu}_q q^{ab} \ptl_b \gamma) \notag\\
&+ \omega' \mathcal{L}_{N'} (N e^{-4\gamma} \bar{\mu}_q q^{ab} \ptl_b \omega) + 2 \mathcal{L}_{N'} N ( \bar{\mu}_q q^{ab} \ptl_b \mbo{\nu}') + 2 \mathcal{L}_{X'} \mbo{\nu}' \bar{\mu}_q q^{ab} \ptl_b N  \notag\\
&- 2 X^a (\bar{\mu}_q q^{bc} \ptl_b \mbo{\nu}' \ptl_c N) + 2 N  q_0^{ac} \varrho'^{b}_c e^{-2\nu} \ptl_b N' + (N \ptl_b N' - N' \ptl_ b N) \tau' \bar{\mu}_q q^{ab}  \notag\\
& -2N' q_0^{ac} \varrho'^{b}_c e^{-2\gamma} \ptl_b N
\end{align} 
after the imposition of the linearly perturbed constraints. In the above, we restricted our discussion to the vacuum (Kerr metric) case because it is directly relevant for our work, but several technicalities related to the asymptotics, global regularity and the related hyperbolic and elliptic theory, in the context of the global Cauchy problem of the more general Kerr-Newman black holes, are addressed comprehensively in \cite{GM17}.
\\

It may be noted that the analogous transformations also resolve the positivity problem for the energy of axially symmetric Maxwell's equations if one does the dimensional reduction to introduce the twist potentials $\lambda, \eta$ corresponding to the $\mathfrak{E}$ and $\mathfrak{B}$ fields (cf. Section 1 in \cite{GM17}). 
\\

In the axially symmetric case, even though the original Maxwell equations are linear, a nonlinear transformation 
is used to reduce the 3+1 Einstein-Maxwell equations to an Einstein-wave map system \cite{kaluz2}, which introduces nonlinear coupling within the Maxwell `twist' fields. However, if we turn off the background $\mathfrak{E}$ and $\mathfrak{B}$ fields (e.g., restrict attention to the Kerr metric), the Maxwell equations in twist potential variables reduce to \emph{linear} hyperbolic PDE. 
 \\
   
Somewhat interestingly, it appears that the construction of a positive-definite energy for the axially symmetric Maxwell equations on Kerr black hole spacetimes does not easily follow from the Carter's identity, but can be realized a special case of the full Robinson's identity. In separate work, Prabhu-Wald have constructed a `canonical energy' for axially symmetric Maxwell equations on Kerr black holes. The connection between the our energy, the Robinson's identity and its generalizations and their canonical energy can be found in Section 1 of \cite{NG_19_1}.
\\

In \cite{NG_17_2}, a positive-definite Hamiltonian energy functional for axially symmetric Maxwell equations propagating on Kerr-de Sitter black hole spacetimes was constructed, using modified Einstein-wave maps for the Lorentzian Einstein manifolds with one rotational isometry \cite{NG_17_1}.
\\

In this work, we shall extend this result and construct a positive-definite energy in a way that is gauge-invariant on the target manifold $(\mathbb{N}, h)$. As we shall see, this is based on negative curvature of the target manifold $(\mathbb{H}^2, h)$ and the convexity of $2+1$ wave maps. The construction of an energy-functional based on the convexity of wave maps, together with our application of the linearization stability methods, suggests why the positivity of our (global) energy for the perturbative theory is to be expected in general, not relying on the insightful and elaborate identity of Carter, which relies on a specific gauge on the target. In such a formulation, the intrinsic geometry within the 2+1 Einstein wave map system becomes more transparent. 
\\

 In the context of black hole uniqueness theorems, extensions along these lines, from  the initial Carter-Robinson results,  were done by Bunting \cite{Bunt_83} and Mazur (\cite{Maz_00} and references therein). In the mathematics literature, convexity of harmonic maps for axially symmetric (Brill) initial data was established by Schoen-Zhou \cite{schoen-zhou_13}, which is often used in geometric inequalities between the area of the horizon, angular-momentum and the mass. 
 \\
 
 In general, due to the geometric nature of the construction, the linearization stability machinery provides a robust mechanism to deal with the stability problems of black holes within a symmetry class, including the initial value problem on hypersurfaces that intersect null infinity. The linearization stability machinery is also equipped for dealing with projections from higher- dimensional $(n+1, n>3)$ black holes with suitable symmetries (toroidal $\mathbb{T}^{n-2}$ spacelike symmetries), including 5D Myers-Perry black holes, the stability of which is the main open problem in the stability of higher dimensional black holes (see e.g., \cite{EmRe_08}). Indeed, most of our current work, especially the local aspects, readily extend to perturbations within the aforementioned symmetry class of higher-dimensional black holes (see below). However, we propose to carefully address the global aspects of this problem, using the methods developed in \cite{GM17}, in a future work. 
\\

We would like to point out there  are related and independent works, based on the `canonical energy' of Hollands-Wald \cite{WH_13}. In \cite{WP_13}, Prabhu-Wald have extended \cite{WH_13} by associating the axisymmetric stability to the existence of a positive-definite `canonical energy'. A positive-definite energy functional was constructed by Dafermos-Holzegel-Rodnianski \cite{HDR_16} in the context of their proof of linear stability of Schwarzschild black holes (see also \cite{GH_16}). Subsequently, a positive-definite energy was constructed by Prabhu-Wald using the canonical energy methods, that is consistent with both \cite{HDR_16} and \cite{Moncrief_74}. Their approach is based on the construction of metric perturbations using the Teukolsky variable as the Hertz potential. A `canonical energy' for perturbations of higher-dimensional black holes within the $\mathbb{T}^{n-2}$ symmetry class, which results in a symplectic structure in the orbit space $\bar{M}/ \mathbb{T}^{n-2},$  is constructed by Hollands-Wald \cite{WH_19}. We would like to point out that their `canonical energy' expression is comparable with the energy constructed in our work (eq. \eqref{e-reg-den} and \eqref{e-reg}). 

From a PDE perspective, a suitable notion of (positive-definite) energy is crucial to control the dynamics of a given system of PDEs. In case the scaling symmetries of a nonlinear hyperbolic PDE and its corresponding energy match, powerful techniques come into play that characterize blow up (concentration) and scattering categorically. This problem is referred to as `energy critical'. In the context of 2+1 critical flat-space wave maps: 

\begin{align}
U \fdg \mathbb{R}^{2+1} \to (\mathbb{N}, h)
\end{align} 

\noindent the fact that this characterization can be made was demonstrated in the landmark works \cite{chris_tah1, jal_tah, jal_tah1, struwe_equi, struwe_sswm, tao_all, krieg_schlag_ccwm,sterb_tata_long, sterb_tata_main} in the analysis of geometric wave equations. It may be noted that 3+1 Einstein’s equations with one translational isometry can be reduced to the 2+1 Einstein-wave map system \eqref{ewm-system}. In this case the notion of a positive-definite, gauge-invariant Hamiltonian mass-energy is provided by Ashtekar-Varadarajan \cite{ash_var} (see also Thorne's C-energy \cite{thorne_cenergy})
\begin{align}
q_{ab} = r^{-m_{AV}} (\delta_{ab} + \mathcal{O}(r^{-1}))
\end{align}
in the asymptotic region of asymptotically flat $(\Sigma, q).$ In a previous work \cite{diss_13}, it was noted that the aforementioned fundamental results on flat space wave maps can be extended to the 2 + 1 Einstein-wave map system resulting from 3 + 1 Einstein’s equations with translational symmetry using the AV-mass, which in turn is related to the energy of 2 + 1 wave maps arising from the energy-momentum tensor and a local conservation law in the equivariant case. We would like to point out that, even though the dimensional reduction of 3 + 1 dimensional axially symmetric,  asymptotically flat spacetimes 
results in the same 2 + 1 Einstein-wave map system locally, the axisymmetric problem is not a (geometric) mass-energy-critical problem \cite{NG17}. There is yet another dimensional reduction, based on the ‘time-translational’ Killing vector of stationary class of spacetimes, in which the Kerr metric also belongs, that results in harmonic maps. This distinction between each of the three cases, which is relevent for the applicable methods therein, is explained in \cite{NG17} for the interested reader.
 \\
 
\noindent Without the energy-criticality of the 2+1 Einstein-wave map system, a direct consideration of the nonlinear problem, analogous to \cite{diss_13, AGS_15}, is infeasible. A long-standing approach that is commonly used in the stability problems of Einstein's equations, is to first consider the linear perturbations and hope to control the nonlinear (higher-order perturbations) using the linear perturbation theory. 
However, the problem of what is the natural notion of energy for the linear perturbative theory, that is consistent with the dimensional reduction and wave map structure of field equations remains open: 

\begin{description}
\item [(\textbf{P2})] Is there a natural notion of mass-energy for the axially symmetric linear peturbative theory of Kerr black hole spacetimes that is consistent with the dimensional reduction and the wave map structure of the equations? 
\end{description}

\noindent This question is closely related to whether there exists a natural factor that multiplies the dimensionally reduced Hamiltonian constraint
of the system (compare with the discussion in pp. 3-4 in \cite{NG17}), which provides a natural notion of energy for our linearized problem. 
We point out that the linearization stability methods employed in our works provide a natural mechanism that kills both the `birds' \textbf{(P1)} $(\vert a \vert < m)$ and \textbf{(P2)} with one `shot', if one may use this terminology. 
\\

\noindent Nevertheless, dealing with a plethora of boundary terms that arise in the construction of the positive-definite energy, in connection with the gauge-conditions  and the dimensional reduction, is nontrivial\footnote{ this is in contrast with the Maxwell perturbations on Kerr black hole spacetimes, which is a (locally) gauge-invariant problem}.  These aspects shall be dealt with in detail in \cite{GM17} for the coupled Kerr-Newman problem. 
\\

\noindent In the current work, after establishing that the constraints for our system are scleronomic, 
we prove that our energy functional drives the  constrained Hamiltonian dynamics of our system 
and  that it forms a (spacetime) divergence-free vector field density, after the imposition of the constraints. 
In the process of obtaining our results, we construct several variational principles  from both Lagrangian and Hamiltonian perspectives, for the nonlinear (exact) and linear perturbative theories. These may be of interest in their own right. 
\\

\noindent As we already remarked, the black hole stability problem is a very active research area. The decay of Maxwell equations on Schwarzschild was proved in \cite{PB_08}. The linear stability of Schwarzschild was established in \cite{HDR_16}. Likewise, the linear stability of Schwarzschild black hole spacetimes using the Cauchy problem for metric coefficients was established in \cite{HKW_16_1, HKW_16_2}, which was recently extended to higher dimensional Schwarzschild-Tangherlini black holes \cite{HKW_18}. A Morawetz estimate for  linearized gravity on Schwarzschild black holes was established in \cite{ABW_17}. These results build on the classic results \cite{Regge-Wheeler_57, Zerilli_70, Zerilli_74, Moncrief_74}. 
\\

The important case of linear wave perturbations of Kerr black holes was studied in several fundamental works for small angular momentum \cite{LB_15_1, Tato_11, DR_11}. Likewise, the decay of Maxwell perturbations of Kerr was established in \cite{LB_15_2}. A uniform energy bound and Morawetz estimate for the  $ \vert s \vert =1, 2$ Teukolsky equations was established in \cite{SMa_17_1,SMa_17_2}. Boundedness and decay for the $\vert s \vert=2$ Teukolsky equation was established in \cite{HDR_17}.  A positive-definite energy for axially symmetric NP-Maxwell scalars was constructed in \cite{NG_19_1}, extending our aforementioned results on Maxwell equations. 
Recently, the linear stability of Kerr black holes was announced in \cite{ABBM_19} by extending the works \cite{SMa_17_1, SMa_17_2} for small $\vert a \vert.$ 
\\

\noindent The effects of the ergo-region become more subtle for rapidly rotating (but $\vert a \vert < M$) Kerr black holes. The decay of the scalar wave for fixed azimuthal modes was established in \cite{FKSY_06, FKSY_08, FKSY_08_E} using spectral methods. The decay of a general linear wave equation was established in \cite{DRS_16}. We would like to point out that the global behaviour, especially the decay estimates, of Maxwell and linearized Einstein perturbations of Kerr black holes,  is relatively less understood for the large, but sub-extremal $(\vert a \vert < M)$ case. We expect that our work will be useful to fill this gap.

\section{A Hamiltonian Formalism for Axially Symmetric Spacetimes}
Recall that $\bar{M} = \olin{\Sigma} \times \mathbb{R} $ is a $3+1$ Lorentzian spacetime, such that the rotational vector field $\Phi$ acts on $\olin{\Sigma}$ as an isometry with the fixed point set $\Gamma.$ In the case of Kerr black hole spacetime, $\Gamma$ is a union of two disjoint sets (the `axes'). It follows that the quotient $\Sigma \fdg= \olin{\Sigma}/SO(2)$ and $M \fdg = \Sigma \times \mathbb{R}$ are manifolds with boundary $\Gamma.$
\noindent Consider the Einstein-Hilbert action on $(\bar{M}, \bar{g})$
\begin{align} \label{EH}
S_{\text{EH}} \fdg = \int \bar{R}_{\bar{g}} \, \bar{\mu}_{\bar{g}}.
\end{align}
Suppose the axially symmetric $(\bar{M}, \bar{g})$ is a critical point of \eqref{EH}. In the Weyl-Papapetrou coordinates,

\begin{align}
\bar{g} = \vert \Phi \vert^{-1} \bar{g} +  \vert \Phi \vert (d \phi + A_\nu dx^\nu)^2
\end{align}
$ \vert \Phi \vert$ is the norm squared of the Killing vector $ \Phi \fdg = \ptl_\phi,$ and $g$ is the metric on the quotient $M \fdg = \bar{M}/SO(2).$ Suppose  $II$ is the second fundamental form of the embedding $(M, \tilde{g}) \hookrightarrow (\bar{M}, \bar{g})$, $\tilde{g} = \vert \Phi \vert^{-1} g$  then following the Gauss-Kodazzi equations and the  conformal transformation, 

\begin{align}
\tilde{R}_{\tilde{g}} = \vert \Phi \vert^{-1} (R_g -4  g^{\mu \nu} \grad_\mu \grad_\nu \log \vert \Phi \vert^{1/2}  - 2 g^{\mu \nu} \grad_\mu \log \vert \Phi \vert^{1/2} \grad_\nu \log \vert \Phi \vert^{1/2} )
\end{align}
The Einstein-Hilbert action \eqref{EH} can be reduced to 
\begin{align}\label{EWM-Lag}
L_{\text{EWM}} \fdg = \halb \int  \left(\frac{1}{\kappa} R_g - h_{AB} (U) g^{\a \b} \ptl_\a U^A \ptl_\b U^B \right) \bar{\mu}_g
\end{align}
for $\kappa =2$ and $U$ is a wave map 
\begin{align}
U \fdg (M, g)  \to (\mathbb{N}, h)
\end{align}
to a hyperbolic 2-plane target $(\mathbb{N}, h)$, whose components are associated to the norm and the twist (potential) of the Killing vector. The tangent bundle of the configuration space of \eqref{EWM-Lag} is now
\begin{align}
C_{\text{EWM}} \fdg = \left \{ (g, \dot{g}), (U^A, \dot{U}^A) \right \}
\end{align}
where the dot (e.g., $\dot{U}$) denotes derivative with respect to a time-coordinate function $t.$
We would like to perform the Hamiltionian reduction of the system \eqref{EWM-Lag}. Recall the ADM decomposition of $(M, g) = (\Sigma, q) \times \mathbb{R}$
\begin{align}
g = -N^2 dt^2 + q_{ab} (dx^a + N^a dt) \otimes (dx^b + N^b dt) 
\end{align}

 Let us split the geometric part and the wave map part of the variational principle \eqref{EWM-Lag} as $L_{\text{EWM}} = L_{\text{geom}} + L_{\text{WM}}.$ 
 Let us now start with the Hamiltonian reduction of the wave map Lagrangian $L_{\text{WM}}$
 \begin{align}\label{WM-lag}
L_{\text{WM}} \fdg= -\halb\int  (h_{AB} (U) g^{\a \b} \ptl_\a U^A \ptl_\b U^B) \bar{\mu}_g
\end{align}
 over the tangent bundle of the configuration space of wave maps, $C_{\text{WM}} = \{ (U^A, \dot{U}^A) \}.$

Suppose we denote the Lagrangian density of \eqref{WM-lag} as $\mathcal{L}$  and conjugate momenta as $p_A,$ we have 

\begin{align}
p_A = \frac{1}{N} \bar{\mu}_q h_{AB}(U)  \ptl_t U^B -  \frac{1}{N} \bar{\mu}_q h_{AB} (U) \mathcal{L}_N  U^B
\end{align}
where $\mathcal{L}_N$ is the Lie derivative with respect to the shift $N^a.$  As a consequence, we have 
\begin{align}
h_{AB} (U)\ptl_t U^B =&\, \frac{1}{\bar{\mu}_q} N p_A + h_{AB} (U) \mathcal{L}_N U^B,
\intertext{ and the Lagrangian density $\mathcal{L}_{\text{WM}}$ can be expressed in terms of the wave map phase space $X_{\text{WM}} \fdg = \{ (U^A, p_A) \}$ as}
\mathcal{L}_{\text{WM}} =&\, \halb p_B \ptl_t U^B - \halb p_B \mathcal{L}_N U^B - \halb  N \bar{\mu}_q h_{AB} (U) q^{ab} \ptl_a U^a \ptl_b U^B. 
\end{align}                    
Let us now define the Hamiltonian density as follows, 
\begin{align}
\mathcal{H}_{\text{WM}} \fdg = \halb p_B \ptl_t U^B + \halb p_B \mathcal{L}_N U^B + \halb  N \bar{\mu}_q h_{AB} (U) q^{ab} \ptl_a U^A \ptl_b U^B. 
\end{align}
As a consequence, we formulate the ADM variational principle for the Hamiltonian dynamics of the wave map phase space $X_{\text{WM}}$ as 

\begin{align}
L_{\text{WM}} [X_{\text{WM}}] \fdg =  \halb \int ( p_A \ptl_t U^A -  p_B \mathcal{L}_N U^B -  N \bar{\mu}_q h_{AB} (U) q^{ab} \ptl_a U^a \ptl_b U^B )d^3x,
\end{align}
which has the field equations, 
\begin{align} \label{u-dot}
p_A = \frac{1}{N} \bar{\mu}_q h_{AB} (U)  \ptl_t U^B -  \frac{1}{N} \bar{\mu}_q h_{AB} (U) \mathcal{L}_N  U^B
\end{align}
and the critical point with respect to  (the first variation $ D_{U^A}  \cdot L_{\text{WM}}=0$) $U^a$ gives 
\begin{align}\label{p-dot}
\ptl_t p_A=& -N \bar{\mu}^{-1}_q \frac{\ptl}{\ptl U^A}h^{BC} p_B p_C + h_{AB} \ptl_a(N \bar{\mu}_q q^{ab} \ptl_b U^B) \notag\\ +   & N\bar{\mu}_q h_{AB} \leftexp{(h)}{\Gamma}^B_{CD}(U)q^{ab}\ptl_a U^C \ptl_b U^D 
 + \mathcal{L}_N p_A,
\intertext{where $\leftexp{(h)}{\Gamma}$ are the Christoffel symbols}
\leftexp{(h)}{\Gamma}^A_{BC} \fdg =&\, \halb h^{AD}(U) (\ptl_C h_{BD} + \ptl_B h_{DC} - \ptl_D h_{BC})
\end{align}
It is straight-forward to verify that the canonical equations

\begin{align}
D_{p_A}\cdot H_{\text{WM}} = \ptl_t U^A \quad \text{and} \quad D_{U^A}\cdot H_{\text{WM}} = - \ptl_t p_A
\end{align}
correspond to \eqref{u-dot} and \eqref{p-dot} respectively,
where $H_{\text{WM}} \fdg = \int \mathcal{H}_{\text{WM}} d^2x $ is the (total) Hamiltonian. Subsequently, if we use the Gauss-Kodazzi equation for the ADM $2+1$ decomposition and defining the (geometric) phase space
\begin{align}
X_{\text{geom}} \fdg = \{ (q_{ab}, \mbo{\pi}^{ab}) \},
\end{align}
 we can represent the gravitational Lagrangian density as follows: 
\begin{align}
\mathcal{L}^{\text{Alt}}_{\text{geom}} \fdg = (- q_{ab}\ptl_t \pi^{ab} - N H_{\text{geom}} - N_a H_{\text{geom}}^a)
\end{align}

where, 
\begin{subequations}
\begin{align}
H_{\text{geom}}\fdg =& \bar{\mu}^{-1}_q  \left( \Vert \mbo{\pi} \Vert_q^2 -  \text{Tr}_q( \mbo{\pi} )^2 \right) - \bar{\mu}_q R_q  \\
H_{\text{geom}}^a \fdg =& -2\, \leftexp{(q)}{\grad}_b \mbo{\pi}^{ab}
\end{align}
\end{subequations}

It may be noted that, the conjugate momentum tensor of the reduced metric and the corresponding components of the conjugate momentum of the $3$ metric are related as follows 

\begin{align}
 \vert \Phi \vert \bar{\mbo{\pi}}^{ab} =   \mbo{\pi}^{ab}, \quad \text{where} \quad \mbo{\pi}^{ab} = \bar{\mu}_q \left(q^{ab} \text{Tr}_q(K) - K^{ab} \right) 
\end{align}

\noindent Consequently, we have the variational principle for Hamiltonian dynamics of the reduced Einstein wave map system. 

\begin{align} \label{ham-var}
J_{\text{EWM}} \fdg =  \int^{t_2}_{t_1}\int_{\Sigma} \left( \mbo{\pi}^{ab} \ptl_t q_{ab} + p_A \ptl_t U^A - N H - N^a H_a \right) d^2x dt 
\end{align}
where now the reduced $H$ and $H_a$ are 
\begin{subequations}
\begin{align}
H= & \bar{\mu}^{-1}_q \left( \left(\Vert \mbo{\pi} \Vert^2_q - \text{Tr}_q (\mbo{\pi})^2 \right) + \halb p_A p^A \right) + \bar{\mu}_q \left(- R_q + \halb h_{AB} q^{ab} \ptl_a U^A  \ptl_b U^B \right) \\
H_a=& - 2\, \leftexp{(q)}{\grad}_b \mbo{\pi}^b_a +  p_A \ptl_a U^A.
\end{align}
\end{subequations}
Therefore, we have proved the theorem
\begin{theorem}
Suppose $(\bar{M}, \bar{g})$ is an axially symmetric, Ricci-flat, globally hyperbolic Lorentzian spacetime and that $\bar{g}$ admits the decomposition \eqref{KK-ADM} in a local coordinate system, then the dimensionally reduced field equations in the interior of $M = \Sigma \times \mathbb{R},$ where $\Sigma = \olin{\Sigma}/SO(2),$ are derivable from the variational principle \eqref{ham-var} for the reduced phase space:
\begin{align}
X_{\textnormal{EWM}} \fdg = \{ (q_{ab}, \mbo{\pi}^{ab}), (U^A, p_A) \}
\end{align} 
\noindent with the Lagrange multipliers $\{ N, N^a \}.$
\end{theorem}
\noindent As a consequence, we have the field equations for Hamiltonian dynamics in $X_{\text{EWM}}$
\begin{subequations}
\begin{align}
h_{AB} \ptl_t U^B =& N \bar{\mu}^{-1}_q p_A + h_{AB} \mathcal{L}_N U^B  \\
\ptl_t p_A=& - N \bar{\mu}^{-1}_q\frac{\ptl}{\ptl U^A}h^{BC} (U) p_B p_C + h_{AB} \ptl_a(N \bar{\mu}_q q^{ab} \ptl_b U^B) \notag\\
&+ N \bar{\mu}_q h_{AB} \leftexp{(h)}{\Gamma}^B_{CD}(U)q^{ab}\ptl_a U^C \ptl_b U^D + \mathcal{L}_N p_A \\
\ptl_t q_{ab} =& 2 N \bar{\mu}^{-1}_q ( \mbo{\pi}_{ab} - q_{ab} \text{Tr} (\mbo{\pi})) + \leftexp{(q)}{\grad}_a N_b + \leftexp{(q)}{\grad}_b N_a \label{qdot}\\
\ptl_t \mbo{\pi}^{ab}=& \halb N \bar{\mu}^{-1}_q q^{ab} (\Vert \mbo{\pi} \Vert^2_q - \text{Tr}(\mbo{\pi})^2)- 2N \bar{\mu}^{-1}_q \left(  \mbo{\pi}^{ac} \mbo{\pi}^{b}_c - \mbo{\pi}^{ab} \text{Tr}(\mbo{\pi}) \right) \notag\\
&+\bar{\mu}_q (\leftexp{(q)}{\grad}^b \,\leftexp{(q)}{\grad}^a N -q^{ab}\, \leftexp{(q)}{\grad}_c \leftexp{(q)}{\grad}^c N ) \notag\\
&+ \leftexp{(q)}{\grad}_c(\mbo{\pi}^{ab} N^c) - \leftexp{(q)}{\grad}_c N^a \mbo{\pi}^{cb} - \leftexp{(q)}{\grad}_c N^b \mbo{\pi}^{ca}\notag\\
&+ \frac{1}{4}\bar{\mu}_q^{-1}N q^{ab} p_A p^A + \halb N \bar{\mu}_q h_{AB} (q^{ac}q^{bd} - \halb q^{ab} q^{cd})\ptl_c U^A \ptl_d U^B 
\end{align}
\end{subequations}
and the constraint equations
\begin{subequations}\label{constraints}
\begin{align} 
H=0, \\
H_a =0.
\end{align}
\end{subequations}
It should be pointed out that, analogous to original ADM formulation \cite{ADM_62}, we have made a simplification with the coupling constant (see also the discussion in pp. 520-521 in \cite{MTW}). In case the precise coupling between the 2+1 Einstein's equations and its wave map source is relevant, the original coupling can be reinstated by simply substituting the following formulas throughout our work:  

\begin{subequations}
\begin{align}
\mbo{\pi}_{\text{true}} \fdg =& \frac{1}{2\kappa} \mbo{\pi}= \frac{1}{2\kappa}  \bar{\mu}_q \left(q^{ab} \text{Tr}(K) - K^{ab} \right),  \\
H_{\text{true}} \fdg=& \frac{1}{2\kappa} H_{\text{geom}} = \bar{\mu}^{-1}_q \frac{1}{2\kappa}  \left( \Vert \mbo{\pi} \Vert_q^2 -  \text{Tr}( \mbo{\pi} )^2 \right) - \frac{1}{2\kappa} \bar{\mu}_q R_q,  \\
(H_{\text{true}})_a \fdg=& \frac{1}{2 \kappa} (H_{\text{geom}})_a = -\frac{1}{\kappa} \leftexp{(q)}{\grad}_b \mbo{\pi}^{b}_a.
\end{align}
\end{subequations}
We would like to remind the reader  that, in the dimensional reduction process, we introduce the closed 1-form $G$ such that 
\begin{align}
\vert \Phi \vert^{-2} \eps_{\mu \nu \b} g^{\b \a} G_\a = F_{\mu \nu}
\end{align}
where $F = dA.$ In our simply connected domain, $G= dw,$ where $\omega$ is the gravitational twist potential and one of the components of the wave map $U.$

\subsection*{Nonlinear Conservation Laws}
Following Komar's definition of angular momentum, 
\begin{align}
J = \frac{1}{16\pi} \int_{\Sigma} \star d \Phi, \quad \text{(Komar angular momentum)}
\end{align}
it follows that for the Kerr metric $J = a M.$ In view of the well known fact that the angular momentum is conserved for our vacuum axisymmetric problem, without effective loss of generality, we shall assume that the perturbation of the angular-momentum is zero.

The dimensional reduction provides additional structure for the original field equations. As noted by Geroch \cite{Geroch_71}, the Lie group $SL(2, \mathbb{R})$ acts on the resulting target $(\mathbb{N}, h)$ in the dimensional reduction procedure. The M\"obius transformations, which are the isometries of $(\mathbb{N}, h)$ provide us a Poisson algebra of nonlinear conserved quantities. 

\begin{corollary}
Suppose $U: \Sigma \times (t_1, t_2 ) \to (\mathbb{N}, h)$ is the wave map coupled to 2+1 Einstein equations as above, 
then there exist  (spacetime) divergence-free vector fields $J_i,\, i =1, 2, 3$ such that if $C_i$ is the flux of $J_i$ at $\Sigma_t, \, t \in (t_1, t_2)$ hypersurface,
\begin{enumerate}
\item 
\begin{align}
\{C_i, C_j \} = \sigma^{k}_{ij} C_k, \quad i\neq j \neq k,
\end{align}
\end{enumerate}
where $\{ \cdot, \cdot \}$ is the Poisson bracket in the phase space $X_{\textnormal{EWM}}$ and  $\sigma^k_{ij}$ are the structure constants of the (M\"obius) isometries $\{ K_i, K_2, K_3\}$ of $(\mathbb{N}, h).$
\end{corollary}
\begin{proof}
The M\"obius transformations on the target $(\mathbb{N}, h)$, the hyperbolic 2-pane, are isometries corresponding to translation, dilation and inversion $\{ K_1, K_2, K_3 \}.$ It follows that, 
\begin{align*}
\leftexp{(h)}{\grad}_A (K_i)_B + \leftexp{(h)}{\grad}_B (K_i)_A =0, \, \forall \, i = 1, 2, 3,
\end{align*}
Consider the quantity
\begin{align}
\ptl_t (K^A_i p_A) =&  \ptl_C K_i^A p_A ( N \bar{\mu}_q p^C + \mathcal{L}_N U^C)  \notag\\
&+ K_i^A \big(-N \bar{\mu}^{-1}_q \ptl_A h^{BC} p_B p_C
+ h_{AB} \ptl_a (N \bar{\mu}_q q^{ab} \ptl_b U^B) \notag\\
&+ N \bar{\mu}_q h_{AB} \leftexp{(h)}{\Gamma}^B_{CD} q^{ab} \ptl_a U^C \ptl_b U^D + \mathcal{L}_N p_A \big).
\end{align}
Now consider,
\begin{align}
\ptl_a (N \bar{\mu}_q q^{ab} \ptl_b U^B K_i^A h_{AB}) =& K_i^A h_{AB} \ptl_a (N \bar{\mu}_q q^{ab} \ptl_b U^B) \notag\\
&+ N \bar{\mu}_q q ^{ab} \ptl_a U^C \ptl_c K_i^A h_{AB} \ptl_b U^B \notag\\
&+ N \bar{\mu}_q q^{ab} \ptl_b U^B K_i^A \ptl_C h_{AB} \ptl_a U^C
\intertext{and note that}
 -N \bar{\mu}_q q ^{ab} \ptl_a U^C \ptl_c K^A h_{AB} \ptl_b U^B - & N \bar{\mu}_q q^{ab} \ptl_b U^B K^A \ptl_C h_{AB} \ptl_a U^C \notag\\
+ N \bar{\mu}_q h_{AB} \leftexp{(h)}{\Gamma}^B_{CD} q^{ab} \ptl_a U^C \ptl_b U^D K_i^A =& -\halb N \bar{\mu}_q q^{ab} \ptl_a U^C \ptl_b U^D \ptl_A h_{CD} \notag\\
&- N \bar{\mu}_q^{ab} \ptl_a U^C \ptl_c U^B \ptl_C K_i^A h_{AB}    \notag\\ 
=&\, 0
\end{align}
after relabeling of indices and on account of the fact that the deformation tensor of $K_i$ in the target $h$ is zero. Let us now define, 
\begin{align*}
(J_i)^t =& \, K^A_i p_A \\
(J_i)^a =& \,  \ptl_a ( N \bar{\mu}_q q^{ab} \ptl_b U^B K_i^A h_{AB} + N^a K_i^A p_A ). 
\end{align*}
It follows from above that each $J_i,\, i = 1, 2, 3$ is a spacetime divergence-free vector density. Now then, 
\begin{align}
C_i \fdg = \int_{\Sigma_t}K_i^A p_A d^2 x
\end{align}

\noindent and consider the Poisson bracket: 
\begin{align}
\{ C_i, C_j \} =&  \Bigg \{ \int_{\Sigma_t} K^A_i p_A \,,\, \int_{\Sigma_t} K^A_j p_A  \Bigg \} \notag\\
=& \int_{\Sigma_t} \left (\ptl_{U^A} K^B_i K^A_j - K^A_i \ptl_{U^A} K^B_j \right) p_B \notag\\
=& \int_{\Sigma_t} \big[K_i, K_j \big]^A p_A = \int_{\Sigma_t} \sigma^k_{ij} K^A_k p_A \notag\\
=& \sigma^k_{ij} C_k, \quad i \neq j \neq k.
\end{align}
This result generalizes the equivalent result in \cite{kaluz1}, where each $\Sigma$ is $\mathbb{S}^2$, to our non-compact case and a general gauge on the target metric $h$. We would like to point out that these conservation laws are closely related to the `moment maps' associated to the M\"obius transformations in the phase space. It may be noted that our arguments readily extend to the $(n+1)$ higher-dimensional context, where the target is $SL (n-2)/SO(n-2).$
\end{proof}

\noindent In the Weyl-Papapetrou coordinates, define a quantity $\mbo{\nu}$ such that 

\begin{align} \label{conformal}
q= e^{2\mbo{\nu}} q_0 
\end{align} 
where $q_0$ is the flat metric and  (mean curvature) scalar $\mbo{\tau} \fdg = \bar{\mu}^{-1}_q q_{ab} \mbo{\pi}^{ab}.$ In our work, it will be convenient to split $2$-tensors into a trace part and the conformal Killing operator \cite{kaluz1}. In the following lemma, we shall streamline the related discussion and results obtained in \cite{kaluz1}.   

\begin{lemma}
Suppose the phase space variables $ \big \{ (q, \mbo{\pi}), (U^A, p_A) \big\} \in X$ are smooth in the interior of $\Sigma$  then, we have 
\begin{enumerate}
\item Suppose  a vector field $Y \in T (\Sigma),$ then conformal Killing operator defined as
\begin{align}
\big(\textnormal{CK}(Y, q)\big)^{ab}\fdg = \bar{\mu}_q ( \leftexp{(q)}{\grad}^b Y^a + \leftexp{(q)}{\grad}^a Y^b - q^{ab} \, \leftexp{(q)}{\grad}_c Y^c) ,
\end{align}
is invariant under a conformal transformation, i.e., $\textnormal{CK} (Y,q) = \textnormal{CK}(Y, q_0).$
\item There exists a vector field $Y \in T(\Sigma),$ which is determined uniquely up a conformal Killing vector,  such that 
\begin{align}
\mbo{\pi}^{ab}=  e^{-2 \mbo{\nu}} (\textnormal{CK}(Y, q))^{ab} + \halb \mbo{\tau} \bar{\mu}_q q^{ab}
\end{align}
\item If  $ \big \{ (q, \mbo{\pi}), (U^A, p_A) \big\} \in X$ satisfy the constraint equations, the Hamilton and momentum constraint equations \eqref{constraints} can be represented as the elliptic equations:
\begin{align}
&\bar{\mu}^{-1}_{q_0} (e^{-2 \mbo{\nu}} \Vert \varrho \Vert^2_{q_0} - \halb \tau^2 e^{2 \mbo{\nu}} \bar{\mu}^2_{q_0} + \halb  p_A p^A) \notag\\
&+ \bar{\mu}_{q_0} (2 \Delta_0 \mbo{\nu} + h_{AB} q_0^{ab} \ptl_a U^A \ptl_b U^B) =0 
\intertext{and}
&-\leftexp{(q_0)}{\grad}_b \varrho^b_a  - \halb \ptl_a \mbo{\tau} e^{2 \mbo{\nu}} \bar{\mu}_{q_0} + \halb p_A \ptl_a U^A=0, \quad (\Sigma, q_0) 
\end{align}
respectively, where 
\begin{align}\label{varrho-def}
\varrho^a_c= \bar{\mu}_{q_0} (\leftexp{(q_0)}{\grad}_c Y^a + \leftexp{(q_0)}{\grad}^a Y_c - \delta^a_c \, \leftexp{(q_0)}{\grad}_b Y^b )
\end{align}
\end{enumerate}
\end{lemma}
\begin{proof}
Part (1) follows from the definitions and direct computations. Part (2) is based on the fact that the transverse-traceless tensors vanish for our form of the 2-metric.  Consider the decomposition of $\mbo{\pi}^{ab}$ into a trace part and a traceless part:
\begin{align} \label{split-pi}
\mbo{\pi}^{ab} =& \halb \tau \bar{\mu}_q q^{ab} + \not \text{Tr}\,\mbo{\pi}^{ab} \notag\\
=& \halb \mbo{\tau} \bar{\mu}_q q^{ab} + (\mbo{\pi}_{\text{TT}})^{ab} + e^{-2\nu} \textnormal{CK} (Y, q) 
\end{align}

\noindent where $(\mbo{\pi}_{\text{TT}})^{ab}$ is such that 

\begin{align} \label{tt-cond}
q_{ab} (\mbo{\pi}_{\text{TT}})^{ab} =0 \quad \text{and} \quad \leftexp{(q)}{\grad}_a (\mbo{\pi}_{\text{TT}})^{ab} =0.
\end{align}

The result (2) now follows from the fact that \eqref{tt-cond} is invariant under the conformal transformation \eqref{conformal} and the fact that transverse-traceless tensors vanish on the flat metric $q_0,$ with suitable boundary conditions. The existence of $Y$ follows from the following elliptic equation 
\begin{align}
\leftexp{(q)}{\grad}_a (e^{-2 \mbo{\nu}} \text{CK} (Y, q)) = \grad_a (\mbo{\pi}^{ab} -\halb \mbo{\tau} \bar{\mu}_q q^{ab}).
\end{align}
 and Fredholm theory. It may be noted that the the right hand side is $L^2-$orthogonal to the kernel of the linear, self-adjoint elliptic operator on the left hand side, which contains the conformal Killing vector fields of $q_0.$  
\begin{align} \label{conf-identity}
&\leftexp{(q)}{\grad}_b \left(e^{-2\mbo{\nu}} \bar{\mu}_q ( Y_a + \leftexp{(q)}{\grad}_a Y^b - \delta^b_a \leftexp{(q)}{\grad}_c Y^c) \right) \notag\\
&= \leftexp{(q_0)}{\grad}_b \left( \bar{\mu}_{q_0}( \leftexp{(q_0)}{\grad}_b Y^a) + \leftexp{(q_0)}{\grad}_a Y^b - \delta^b_a \leftexp{(q_0)}{\grad}_c Y^c ) \right).
\end{align}
It would now be convenient to define $\varrho$ as in \eqref{varrho-def}. Now then, using 
\begin{align}
\mbo{\pi}^a_b = \halb \mbo{\tau} e^{2 \mbo{\nu}} \bar{\mu}_{q_0} \delta^a_c + e^{-2 \mbo{\nu}} \varrho^a_c,
\end{align}
and \eqref{conf-identity}, the momentum constraint can now be transformed into the following elliptic operator for $Y$ on $(\Sigma, q_0)$
\begin{align}
H_a = -2\leftexp{(q_0)}{\grad}_b \varrho^b_a  - \ptl_a \mbo{\tau} e^{2 \mbo{\nu}} \bar{\mu}_{q_0} +  p_A \ptl_a U^A, \quad (\Sigma, q_0), \quad a= 1, 2.
\end{align}
The scalar curvature $R_q$ of $(\Sigma, q)$ and $\Vert \mbo{\pi} \Vert^2_q$ can be expressed explicitly as 
\begin{align}
R_q = -2 e^{-2 \mbo{\nu}} \Delta_0 \mbo{\nu}, \quad \Vert \mbo{\pi} \Vert^2_q = \halb \mbo{\tau}^2 e^{4 \nu} \bar{\mu}^2_{q_0} + \Vert  \varrho \Vert^2_{q_0}.
\end{align}
The Hamiltonian constraint can now be transformed to the elliptic operator 
\begin{align}
H =& \bar{\mu}^{-1}_{q_0} (e^{-2 \mbo{\nu}} \Vert \varrho \Vert^2_{q_0} - \halb \mbo{\tau}^2 e^{2 \mbo{\nu}} \bar{\mu}^2_{q_0} + \halb  p_A p^A) \notag\\
&+ \bar{\mu}_{q_0} (2  \Delta_0 \mbo{\nu} + \halb h_{AB} q_0^{ab} \ptl_a U^A \ptl_b U^B), \quad (\Sigma, q_0), 
\intertext{where}
\Delta_0 \mbo{\nu} \fdg=& \frac{1}{\bar{\mu}_{q_0}} \ptl_b(q^{ab}_0 \bar{\mu}_{q_0} \ptl_b \mbo{\nu}).
\end{align}
\end{proof}
The conditions for the dimensional reduction above are modeled along the Kerr metric \eqref{BL-Kerr}. 
Let us now consider the corresponding field equations for the Kerr metric. It follows that, for the Kerr wave map 
\begin{align}
 U \fdg (M, g) \to (N, h)
 \end{align}
 we have $p_1= p_2 \equiv 0;$  and $ \mbo{\pi}^{ab}  \equiv 0.$ As a consequence, the dimensionally reduced field equations for the Kerr metric \eqref{BL-Kerr} are
 \begin{align}
 \ptl_a ( N \bar{\mu}_q q^{ab} U^A) + N \bar{\mu}_q \leftexp{(h)}{\Gamma}^{A}_{BC} q^{ab} \ptl_a U^B \ptl_b U^C=&0 \label{Kerr-p-dot}
 \intertext{and}
 \bar{\mu}_q \left ( \leftexp{(q)}{\grad}^b \, \leftexp{(q)}{\grad}^a N - q^{ab} \, \leftexp{(q)}{\grad}_c \leftexp{(q)}{\grad}^c N \right) &\notag\\
 + \halb N \bar{\mu}_q (q^{ac} q^{bd} - \halb q^{ab} q^{cd}) h_{AB} \ptl_a U^A \ptl_b U^B=&0.
 \end{align}
\noindent The Hamiltonian constraint
\begin{align}
H = \bar{\mu}_q (-R_q  + \halb h_{AB} q^{ab} \ptl_a U^A \ptl_b U^B) =0
\end{align}
for $a, b$ and $A, B, C = 1, 2.$
The scalar $\mbo{\tau}$ is the mean curvature of the embedding $\Sigma \hookrightarrow M,$ whose evolution is governed by the equation:
\begin{align}\label{meancurv}
\ptl_t \mbo{\tau} = - \leftexp{(q)}{\grad}_a \leftexp{(q)}{\grad}^a N + N \bar{\mu}^{-1}_q ( \Vert \mbo{\pi} \Vert_q^2 + \halb p_A p^A).
\end{align}
Following the notation introduced in \cite{kaluz1}, the evolution equation \eqref{meancurv} can be represented as 
\begin{align}
e^{2\mbo{\nu}} \ptl_t \mbo{\tau} = - \Delta_0 N + N \mathfrak{q},
\end{align}
where 
\begin{align}
\mathfrak{q} \fdg= \bar{\mu}^{-1}_q e^{-2\mbo{\nu}} \left( \Vert \varrho \Vert^2_q + \halb \mbo{\tau}^2 e^{4\mbo{\nu}} \bar{\mu}_q +  \halb p_A p^A \right)
\end{align}
where we again used the splitting expression \eqref{split-pi}. It follows that, for the Kerr metric \eqref{BL-Kerr}, 
\begin{align}
\mbo{\tau} = \ptl_t \mbo{\tau} \equiv 0, \quad \varrho \equiv 0 \quad \text{and} \quad \Delta_q N =0. \label{kerr-maximal}
\end{align} 
The equation \eqref{qdot} can be decomposed as 
\begin{align}
\ptl_t (\bar{\mu}_q) =&  N \text{Tr}_q (\mbo{\pi}) - \halb \bar{\mu}_q q_{ab}\left(\leftexp{(q)}{\grad}^a N^b + \leftexp{(q)}{\grad}^b N^a \right)
\intertext{and the evolution of the densitized inverse metric}
\ptl_t (\bar{\mu}_q q^{ab}) =& 2 N (\mbo{\pi}^{ab} - \halb q^{ab} \text{Tr}_q(\mbo{\pi})) + \bar{\mu}_q \big( \leftexp{(q)}{\grad}^a N^b + \leftexp{(q)}{\grad}^b N^a - q^{ab} \leftexp{(q)}{\grad}_c N^c \big).
\end{align}

\section{A Hamiltonian Formalism for Axially Symmetric Metric Perturbations}
In this section we shall calculate the field equations and the Lagrangian and Hamiltonian variational principles for linear perturbation equations of the 2+1 Einstein-wave map system. Consider a smooth curve 
\begin{align}
\mbo{\gamma}_s \fdg [0, 1] \to C_{\text{EWM}}
\end{align}
parametrized by $s$ in the tangent bundle of configuration space $C_{\text{EWM}}$ of the Einstein-wave map system. Like previously, we shall start with the wave map system. Let $U_s \fdg (M, g) \to (N, h)$ be  a 1-parameter family of maps generated by the flow along $\mbo{\gamma}_s,$ such that 
\begin{subequations}
\begin{align}
U_0 \equiv& \, U \\
U_s \equiv& \, U, \quad \textnormal{outside a compact set $\Omega \subset M$} 
\end{align}
\end{subequations}
and $ U' \fdg = D_{\mbo{\gamma}_s} \cdot U_s \big\vert_{s=0},$ where $U \fdg (M, g) \to (N, h)$ is a given (e.g., Kerr) wave map. The deformations along $\mbo{\gamma}_s$ can be manifested, for instance, by the exponential map $\text{Exp}(s U)$. In the following, with a slight abuse of notation, we shall denote the manifestations of the deformations along $\mbo{\gamma}_s$ for the wave map $U \fdg (M,g) \to (N, h )$, by $\mbo{\gamma}_s$ itself.   Let us now denote the deformations along $\mbo{\gamma}_s$ of a point at $s=0$ in the tangent bundle of the wave map configuration space $C_{\text{WM}}$ as follows 
\begin{align}
C'_{ \text{WM}} \fdg = \Big\{
U^{'A} = D_{\mbo{\gamma}_s} \cdot U_s^{A} \Big \vert_{s=0}, \quad \dot{U}^{'A} =  D_{\mbo{\gamma}_s} \cdot \dot{U}_s^{A} (s) \Big \vert_{s=0} \Big\}.
\end{align}

 \noindent Now consider the Lagrangian action of wave map
\begin{align} \label{wm-action-again}
L_{\text{WM}} (C_{\text{WM}})  =  -\halb \int (g^{\mu \nu}h_{AB} \ptl_\mu U^A \ptl_\nu U^B ) \bar{\mu}_g.
\end{align}
For simplicity, we shall denote $L_{\text{WM}}(\gamma(s))$ as $L_{\text{WM}}(s).$ We have, 

\begin{align}\label{1-var-wm}
D_{\mbo{\gamma}_s} \cdot L_{\text{WM}}(s) =   \int h_{AB} (\square_g U^A + \leftexp{(h)}{\Gamma}^A_{BC} g^{\mu \nu} \ptl_\mu U^B \ptl_\nu U^C) U^{'B} \bar{\mu}_g 
\end{align}
where we have used the identity, 
\begin{align} \label{christof-intro}
& g^{\mu \nu} h_{AB}(U) \ptl_\mu U^A \ptl_\nu U^{'B} + \halb g^{\mu \nu}\ptl_C h_{AB} \ptl_\mu U^A \ptl_\nu U^B U^{'C} \notag\\
& = - h_{AB} U^{'B} (\square_g U^A + \leftexp{(h)}{\Gamma}^A_{BC} g^{\mu \nu} \ptl_\mu U^B \ptl_\nu U^C )
\end{align}
modulo boundary terms (see e.g., pp. 19-20 in \cite{diss_13}). The following geometric construction shall be useful to represent our formulas compactly \cite{Misn_78}. Firstly, let us define the notions of induced tangent bundle and the associated `total' covariant derivative on the target $(\mathbb{N}, h),$ under the wave mapping $U \fdg M \to N.$ The induced tangent bundle $T_U N$ on $M$ consists of the 2-tuple $(x, y),$ where $x \in M$ and $y \in T_{U(x)} N,$ with the bundle projection 

\begin{align}
P \fdg T_U N &\to M \notag\\
(x, y) &\to x.
\end{align}
Consider the vector field $\dot{V}_s \in TM,$ then the image of $\dot{V}_s$ under the wave map $U$ is a vector field $\dot{V}_s^A = \ptl_s U^A$ in a local coordinate system of $(N, h).$ As a consequence, we can define a covariant derivative on the induced bundle:
 
\begin{align}
\leftexp{(h)}{\grad}_\mu \dot{V}^A_s \fdg= \ptl_\mu \dot{V}_s^A +   \leftexp{(h)}{\Gamma}^A_{BC} \dot{V}_s^B \ptl_\mu U^C.
\end{align} 
It may be verified explicitly that the induced connection is metric compatible $\leftexp{(h)}{\grad}_A h^{AB} \equiv 0.$ Likewise, for a `mixed' tensor $$\Lambda \fdg= \Lambda^A_\mu \, \ptl_{x^A} \otimes  dx^\mu,$$

\begin{align}
\leftexp{(h)}{\grad}_\nu \Lambda_\mu^A \fdg= \leftexp{(g)}{\grad}_\nu \Lambda_\mu^A + \leftexp{(h)}{\Gamma}^A_{BC} \Lambda^B_\mu \ptl_\nu U^C.
\end{align} 
In particular, for $\mbo{e}_B \in TN,$ the second covariant derivative, 
\begin{align}
\leftexp{(h)}{\grad}_ \mu \leftexp{(h)}{\grad}_\nu \mbo{e}_B = \ptl_\mu ( \Gamma^A_{\nu B} \mbo{e}_A) - \leftexp{(g)}{\Gamma}^\a_{\mu \nu} \leftexp{(h)}{\Gamma}^A_{\a B} \mbo{e}_A + \leftexp{(h)}{\Gamma}^A_{\mu B} \leftexp{(h)}{\Gamma}^C_{\nu A} \mbo{e}_C
\end{align} 
provides the curvature for the induced connection: 
\begin{align}
\big[ \grad_\mu, \grad_\nu \big] \mbo{e}_B = R^A_{\,\,\,\, B \mu \nu} \, \mbo{e}_A
\end{align}

\noindent Now consider the `mixed' second covariant derivatives 
\[ \leftexp{(h)}{\grad}_\mu \leftexp{(h)}{\grad}_A \mbo{e}_B \quad \text{and} \quad \leftexp{(h)}{\grad}_A \leftexp{(h)}{\grad}_\mu  \mbo{e}_B. \] In view of the fact that $\mbo{e}_B$ and $\leftexp{(h)}{\grad}_A \mbo{e}_B $ do not have components in the tangent bundle of the domain $M,$ the quantities \[ \ptl_\mu U^C\leftexp{(h)}{\grad}_C \leftexp{(h)}{\grad}_A \mbo{e}_B \quad \text{and} \quad \leftexp{(h)}{\grad}_A ( \ptl_\mu U^C \leftexp{(h)}{\grad}_C) \mbo{e}_B \] are equivalent to \[ \leftexp{(h)}{\grad}_\mu \leftexp{(h)}{\grad}_A \mbo{e}_B \quad \text{and} \quad \leftexp{(h)}{\grad}_A \leftexp{(h)}{\grad}_\mu \mbo{e}_B \quad \text{respectively}.\] 
We have,
\begin{subequations} \label{mixed-curv}
\begin{align}
& U'^A \, \leftexp{(h)}{\grad}_A \leftexp{(h)}{\grad}_\mu \mbo{e}_B = U'^A \, \leftexp{(h)}{\grad}_A \left(\ptl_\mu U^C \, \leftexp{(h)}{\grad}_C \mbo{e}_B \right)  \notag\\
& \quad = U'^A \ptl_\mu U^C (\ptl_A \leftexp{(h)}{\Gamma}^D_{CB} + \leftexp{(h)}{\Gamma}^D_{AE} \leftexp{(h)}{\Gamma}^E_{CB} ) +
 U'^A \leftexp{(h)}{\grad}_A \ptl_\mu U^C \leftexp{(h)}{\grad}_C \mbo{e}_B 
 \intertext{likewise}
 & \ptl_\mu U^A \, \leftexp{(h)}{\grad}_A (U'^C\leftexp{(h)}{\grad}_C \mbo{e}_B) \notag\\
 &= \, \ptl_\mu U^A U'^C \leftexp{(h)}{\grad}_A \leftexp{(h)}{\grad}_C  \mbo{e}_B + \ptl_\mu U^A \leftexp{(h)}{\grad}_A U'^C \leftexp{(h)}{\grad}_C \mbo{e}_B \notag\\
 & = \, \ptl_\mu U^A U'^C (\ptl_C \leftexp{(h)}{\Gamma}^D_{AB} + \leftexp{(h)}{\Gamma}^D_{CE} \leftexp{(h)}{\Gamma}^E_{AB} ) 
 +  \ptl_\mu U^A \leftexp{(h)}{\grad}_A U'^C \leftexp{(h)}{\grad}_C \mbo{e}_B 
\end{align}
\end{subequations}
so that we have 
\begin{align}
U'^A \, \leftexp{(h)}{\grad}_A \leftexp{(h)}{\grad}_\mu \mbo{e}_B - \ptl_\mu U^A \, \leftexp{(h)}{\grad}_A (U'^C\leftexp{(h)}{\grad}_C \mbo{e}_B) = \leftexp{(h)}{R}^D_{\,\,\,\, B A \mu} \mbo{e}_D U'^A.
\end{align}

\noindent This `mixed' derivative construction is relevant for our wave map deformations. 
\noindent  Let us assume that

\begin{subequations} \label{s-t-coordinates}
 \begin{align}
 \big[ \ptl_\b U, U' \big] \equiv&\, 0
 \intertext{from which, it follows that} 
 \ptl_\b U^A \,\, \leftexp{(h)}{\grad}_A U'^B -  U'^A\,\,\leftexp{(h)}{\grad}_A \ptl_\b U^B \equiv&\, 0.
 \end{align}
 \end{subequations}
Now consider another analogous curve $\mbo{\gamma}_\lambda.$  The quantity $D^2_{\mbo{\gamma}_\lambda \mbo{\gamma}_s} \cdot L_{\text{WM}} $
involves the following terms

\begin{align}\label{2-var-start}
\square_g U'^{A} + \ptl_{U^D} \leftexp{(h)}{\Gamma}^A_{B C} g^{\mu \nu} \ptl_\mu U^B \ptl_\nu U^C U'^D + 2 \leftexp{(h)}{\Gamma}^A_{BC} g^{\mu \nu} \ptl_\mu U'^B \ptl_\nu U^C.
\end{align}
Assuming that the Kerr wave map is a critical point of \eqref{1-var-wm} at $s=0,$ the expression \eqref{wm-action-again} can consecutively be transformed as follows: 
\begin{align}
=& \, \square_g U'^{A} + \ptl_{U^D} \leftexp{(h)}{\Gamma}^A_{B C} g^{\mu \nu} \ptl_\mu U^B \ptl_\nu U^C U'^D + 2 \leftexp{(h)}{\Gamma}^A_{BC} g^{\mu \nu} \ptl_\mu U'^B \ptl_\nu U^C \notag\\
&\quad + \leftexp{(h)}{\Gamma}^{A}_{BC} U'^B (\square_g U^C + \leftexp{(h)}{\Gamma}^{C}_{DE} g^{\a\b} \ptl_\a U^D \ptl_\b U^E ) \notag\\
\end{align}
which can be transformed to 
\begin{align} \label{curv-term1}
g^{\mu \nu} U'^C \, \leftexp{(h)}{\grad}_C \leftexp{(h)}{\grad}_\mu \ptl_\nu U^A =&  g^{\mu \nu} U'^C \Big( \ptl_C  (\leftexp{(h)}{\grad}_\mu \ptl_\nu U^A) - \leftexp{(h)}{\Gamma}^D_{C \mu} \leftexp{(h)}{\grad}_D \ptl_\nu U^A  \notag\\
& + \leftexp{(h)}{\Gamma}^{A}_{CD} \leftexp{(h)}{\grad}_\mu \ptl_\nu U^A \Big)
\end{align}
Now consider the operator, 
\begin{align} \label{curv-term2}
g^{\mu\nu}\leftexp{(h)}{\grad}_\mu (\leftexp{(h)}{\grad}_C \ptl_\nu U^A) =& g^{\mu \nu} \leftexp{(g)}{\grad}_\mu (\leftexp{(h)}{\grad}_C \ptl_\nu U^A) \notag\\
&- g^{\mu \nu}  \left( \leftexp{(h)}{\Gamma}^D_{\mu C} \grad_D \ptl_\nu U^A + \leftexp{h}{\Gamma}^A_{\mu D} \leftexp{(h)}{\grad}_C \ptl_\nu U^D \right)
\end{align}
and performing the computations analogous to \eqref{mixed-curv}, we get that \eqref{2-var-start} is equivalent to 
\begin{align}
\leftexp{(h)}{\square}\, U'^A + \leftexp{(h)}{R}^A_{\,\,\,\,BCD} g^{\mu \nu} \ptl_\mu U^B \ptl_\nu U^D U'^C
\end{align}

where 
\begin{align}
\leftexp{(h)}{\square} \, U'^A \fdg =& g^{\mu \nu} \leftexp{(h)}{\grad}_\mu \leftexp{(h)}{\grad}_\nu U'^A  \notag\\
=&  \ptl_\mu ( \leftexp{(h)}{\grad}_\nu U'^A) - \leftexp{(g)}{\Gamma}^{\gamma}_{\mu \nu} \leftexp{(h)}{\grad}_\gamma U'^A + \leftexp{(h)}{\Gamma}^A_{\mu C} (\leftexp{(h)}{\grad}_\nu U'^C),
\intertext{which can be represented in terms of the covariant wave operator $(g^{\mu \nu}\leftexp{(g)}{\grad}_{\mu} \ptl_\nu U'^A )$ in the domain metric $g$ as}
=& \square_g U'^A +  g^{\mu \nu} \big(\ptl_\mu (\leftexp{(h)}{\Gamma}^A_{\nu C} U'^C) - \leftexp{(g)}{\Gamma}^\gamma_{\mu \nu} \leftexp{(h)}{\Gamma}^A_{\gamma C} U'^C + \leftexp{(h)}{\Gamma}^A_{\mu C} \ptl_\nu U'^C  \notag\\
&+ \leftexp{(h)}{\Gamma}^A_{\mu C} \leftexp{(h)}{\Gamma}^C_{\nu D} U'^D \big)
\end{align}

\begin{align}
\intertext{and $\leftexp{(h)}{R}$ is the induced Riemannian curvature tensor}
 \leftexp{(h)}{R}^A_{\,\,\,\,BCD} =& \ptl_C\leftexp{(h)}{\Gamma}^A_{DB} - \ptl_D\leftexp{(h)}{\Gamma}^A_{CB} + \leftexp{(h)}{\Gamma}^A_{CE} \leftexp{(h)}{\Gamma}^E_{DB} - \leftexp{(h)}{\Gamma}^A_{DE}\leftexp{(h)}{\Gamma}^E_{CB}
\end{align}

\noindent Now for the Kerr wave  map critical point of $D_{\mbo{\gamma}_s} \cdot L_{\text{WM}}$ at $s=0,$ we then have, 
\begin{align} \label{div-var}
D^2_{\mbo{\gamma}_\lambda \mbo{\gamma}_s} \cdot L_{\text{WM}} (s=0) = \int h_{AB} U'^B ( \leftexp{(h)}{\square} U'^A + \leftexp{(h)}{R}^A_{\,\,\,\,BCD} g^{\mu \nu} \ptl_\mu U^B \ptl_\nu U^D U'^C) \bar{\mu}_g
\end{align}
as the Lagrangian variational principle for small linear deformations of the wave map $U_s \fdg (M, g) \to (N, h).$ 
In view of the divergence identity,
\begin{align}
\leftexp{(h)}{\grad}_\mu (h_{AB} U'^B \, \leftexp{(h)}{\grad}^\mu U'^A)= h_{AB} \leftexp{(h)}{\grad}^\mu U'^A \leftexp{(h)}{\grad}^{\mu} U'^B + 
U'^B \leftexp{(h)}{\grad}_{\mu} ( h_{AB} \leftexp{(h)}{\grad}^\mu U'^A)
\end{align}
the variational principle \eqref{div-var} can equivalently be transformed into a self-adjoint variational form: 

\begin{align}\label{quad-var}
& D^2_{\mbo{\gamma}_\lambda \mbo{\gamma}_s} \cdot L_{\text{WM}} (s=0) \notag\\
&= - \halb  \int \big( g^{\mu \nu}\, h_{AB}\leftexp{(h)}{\grad}_\mu U'^A \leftexp{(h)}{\grad}_\mu U'^B - h_{AB} U'^B \, \leftexp{(h)}{R}^A_{BCD} g^{\mu \nu} \ptl_\mu U^B \ptl_\nu U^C U'^D \big) \bar{\mu}_g. 
\end{align}
Let us now calculate the Hamiltonian field equations for the linear perturbation theory, using the ADM decomposition of the background $(M, g)$
\begin{align}
g = -N^2 dt^2 + q_{ij} (dx^i + N^i dt) \otimes (dx^j + N^j dt)
\end{align}
Let us denote the variational principle \eqref{div-var} and \eqref{quad-var} by $L_{\text{WM}}(U').$ The Legendre transformation on $C'_{\text{WM}}$ results in the phase space 
\begin{align}
X'_{\text{WM}} \fdg = \{ (U'^{A}, p'_A) \}, \quad \text{where $(U'^A, p'_A)$ are canonical pairs}
\end{align}
the conjugate momenta $p'_A =  D_{\mbo{\gamma}} \cdot ( p_{A} (s)) \big \vert_{s=0}$  are given by 
\begin{align} \label{u-prime-dot-2}
p'_A=& \frac{1}{N} \bar{\mu}_q h_{AB}(U) (\ptl_t U'^B + \leftexp{(h)}{\Gamma}^B_{t C} U'^C) - \frac{\bar{\mu}_q}{N} h_{AB}(U) \mathcal{L}_N U'^B \notag\\
&- \frac{\bar{\mu}_q}{N} h_{AB}(U) N^a \leftexp{(h)}{\Gamma}^B_{a C} 
\end{align} 
on account of the fact that the time derivative terms in the second term of \eqref{quad-var} only occur for background wave map $U.$ Now then, using the quantity
\begin{align}
h_{AB} \ptl_t U'^B = \bar{\mu}^{-1}_q N p'_A - h_{AB}(U)\leftexp{(h)}{\Gamma}^B_{tC} U'^C + h_{AB} \mathcal{L}_N U'^B + h_{AB} N^a \leftexp{(h)}{\Gamma}^B_{a C} U'^C,
\end{align}
the Lagrangian and Hamiltonian densities, $\mathcal{L}'_{\text{WM}}$ and $\mathcal{H}'_{\text{WM}},$ can be expressed in terms of the phase space variables $X'_{\text{WM}} = \{ (U'^A, p'_A) \},$ in a recognizable ADM form as follows 

\begin{subequations}
\begin{align}
\mathcal{L}'_{\text{WM}} (U') \fdg= 
&  \halb p'_A \ptl_t U'^A  - \halb p'_A \mathcal{L}_N U'^A + \left(\leftexp{(h)}{\Gamma}^A_{tC} U'^C - \leftexp{(h)}{\Gamma}^A_{a C} N^a U'^C \right)\halb p'_A \notag\\
&-\halb h_{AB}(U) N \bar{\mu}_q q^{ab} \left(\,\leftexp{(h)}{\grad}_a U^A \, \leftexp{(h)}{\grad}_b U^B \right) \notag\\
&+ N \bar{\mu}_q h_{AE}(U) U'^A   R^E_{\,\,\,\,BCD} q^{a b} \ptl_a U^B U'^C \ptl_b U^D \notag\\
& -\frac{1}{N} \bar{\mu}_q h_{AE}(U) U'^A   R^E_{\,\,\,\,BCD} \,  \mathcal{L}_N U^B U'^C \mathcal{L}_N U^D \notag\\
&- \frac{1}{N} \bar{\mu}_q h_{AE}(U) U'^A   R^E_{\,\,\,\,BCD}  \ptl_t U^B U'^C \ptl_t U^D \notag\\
& + \frac{2}{N}\bar{\mu}_q h_{AE}(U) U'^A   R^E_{\,\,\,\,BCD} q^{a b} \ptl_t U^B U'^C \mathcal{L}_N U^D
\intertext{likewise the Hamiltonian energy density can be expressed as,}
\mathcal{H}'_{\text{WM}} 
 \fdg=&  \halb p'_A \ptl_t U'^A  + \halb p'_A \mathcal{L}_N U'^A - \left(\leftexp{(h)}{\Gamma}^A_{tC} U'^C - \leftexp{(h)}{\Gamma}^A_{a C} N^a U'^C \right)\halb p'_A \notag\\
&+\halb h_{AB}(U) N \bar{\mu}_q q^{ab} \left(\,\leftexp{(h)}{\grad}_a U^A \, \leftexp{(h)}{\grad}_b U^B \right) \notag\\
& - N \bar{\mu}_q h_{AE}(U) U'^A   R^E_{\,\,\,\,BCD} q^{a b} \ptl_a U^B U'^C \ptl_b U^D \notag\\
& +\frac{1}{N} \bar{\mu}_q h_{AE}(U) U'^A   R^E_{\,\,\,\,BCD} \,  \mathcal{L}_N U^B U'^C \mathcal{L}_N U^D \notag\\
& + \frac{1}{N} \bar{\mu}_q h_{AE}(U) U'^A   R^E_{\,\,\,\,BCD}  \ptl_t U^B U'^C \ptl_t U^D \notag\\
& - \frac{2}{N}\bar{\mu}_q h_{AE}(U) U'^A   R^E_{\,\,\,\,BCD} q^{a b} \ptl_t U^B U'^C \mathcal{L}_N U^D
\end{align}
\end{subequations}
so that the critical point of $L'_{\text{WM}}$

\begin{align}
L'_{\text{WM}} = \int^{t_2}_{t_1} \int_{\Sigma} \mathcal{L}'_{\text{WM}} \,\, d^2x dt
\end{align}

\noindent with respect to $U'^A$ gives the field equation

\begin{align} \label{p-prime-dot-2}
\ptl_t p'_A =& \mathcal{L}_N p'_A + (\leftexp{(h)}{\Gamma} ^C_{tA} - \leftexp{(h)}{\Gamma}^C_{aA} N^a )p'_C + h_{AB} \leftexp{(h)}{\grad}_a ( N \bar{\mu}_q q^{ab} \grad_b  U'^B) \notag\\
& + N \bar{\mu}_q h_{AE}(U)  R^E_{\,\,\,\,BCD} q^{a b} \ptl_a U^B U'^C \ptl_b U^D \notag\\
& -\frac{1}{N} \bar{\mu}_q h_{AE}(U)   R^E_{\,\,\,\,BCD} \,  \mathcal{L}_N U^B U'^C \mathcal{L}_N U^D \notag\\
&- \frac{1}{N} \bar{\mu}_q h_{AE}(U)    R^E_{\,\,\,\,BCD}  \ptl_t U^B U'^C \ptl_t U^D \notag\\
& + \frac{2}{N}\bar{\mu}_q h_{AE}(U)   R^E_{\,\,\,\,BCD} q^{a b} \ptl_t U^B U'^C \mathcal{L}_N U^D.
\end{align}


\noindent Analogously, it is straightforward to note that the field equations \eqref{u-prime-dot-2} and \eqref{p-prime-dot-2} are generated by the Hamiltonian $H'_{\text{WM}} = \int \mathcal{H}'_{\text{WM}} d^2 x,$ i.e., 

\begin{align}
D_{p'^A} \cdot H'_{\text{WM}} =  \ptl_t U'^A, \quad D_{U'^A} \cdot H'_{\text{WM}} = - \ptl_t p'_A,
\end{align}
respectively. Specializing to our stationary Kerr background metric, we have 

\begin{subequations}
\begin{align}
 h_{AB}(U) \ptl_t U'^B=& \bar{\mu}^{-1}_q N p'_A, \\
\ptl_t p'_A =& h_{AB} (U) \leftexp{(h)}{\grad}_a ( N \bar{\mu}_q q^{ab} \grad_b  U'^B) \notag \\
& \quad \quad + N \bar{\mu}_q h_{AE}(U)  R^E_{\,\,\,\,BCD} q^{a b} \ptl_a U^B U'^C \ptl_b U^D
\end{align}
\end{subequations}

 Let us now construct the variational principle for the fully coupled Einstein-wave map perturbations. 
Now suppose,
\begin{align}
q'_{ab} =  D_{\mbo{\gamma}_s} \cdot (q_{ab} (s)) \Big \vert_{s=0}, \quad \mbo{\pi}'_{ab} =  D_{\mbo{\gamma}_s} \cdot (\mbo{\pi}_{ab} (s)) \Big \vert_{s=0},
\end{align}

\noindent let us then denote the phase space corresponding to the perturbative theory of Kerr metric as $X'_{\text{EWM}}:$ 
\begin{align}
X' \fdg = \big \{ (U'^{A}, p'_A), (q'_{ab}, \mbo{\pi}'_{ab}) \big \}.
\end{align} 

\noindent Using the gauge-condition that the densitized metric $\bar{\mu}^{-1}_q q_{ab}$ is fixed, we can construct $D_{\mbo{\gamma}_s} \cdot H $ and $D_{\mbo{\gamma}_s} \cdot H_a $ at $s=0$

\begin{align}
H' \fdg =& D_{\mbo{\gamma}_s} \cdot H (s=0) = - \bar{\mu}^{-1}_q q_{ab} \mbo{\pi}'^{ab}  - (\bar{\mu}_q R_q)' \notag\\
&+  \halb \bar{\mu}_q q^{ab} \ptl_{U^C} h_{AB} (U) \ptl_a U^A \ptl_b U^B U'^C 
\bar{\mu}_q q^{ab} h_{AB}(U)  \ptl_a U'^A \ptl_b U^B \label{H'}
\intertext{and}
H'_a \fdg=& D_{\mbo{\gamma}_s} \cdot H_a (s=0) = \leftexp{(q)}{\grad}_b \mbo{\pi}'^b_a + p'_A U^A,
\intertext{where}
(\bar{\mu}_q R_q)' =& \bar{\mu}_q \left( - \Delta_q q' + \leftexp{(q)}{\grad}^a \leftexp{(q)}{\grad}^b q'_{ab}  \right),\quad q' \fdg = \text{Tr}_q q'_{ab}.
\end{align}

\noindent Again, after imposing that the Kerr metric is a critical point at $s=0,$ we get

\begin{align}
& D^2_{\mbo{\gamma}_\lambda \mbo{\gamma}_s} \cdot H(s=0) \notag\\
&= \bar{\mu}^{-1}_q ( 2 \Vert \mbo{\pi}' \Vert_q^2 -  2(q_{ab} \mbo{\pi}'^{ab})^2 + p'_A p'^A )  
\notag\\ 
& \quad -  (\bar{\mu}_q R_q)'' + \halb \bar{\mu}_q q^{ab} \ptl^2_{U^D U^C} h_{AB}(U) \ptl_a U^A \ptl_b U^B U'^C U'^D \notag\\
& \quad + \bar{\mu}_q q^{ab} \ptl_{U^C} h_{AB}(U) \ptl_a U'^A \ptl_b U^B U'^C \notag\\
& \quad + \halb \bar{\mu}_q q^{ab} \ptl_{U^C} h_{AB}(U) \ptl_a U^A \ptl_b U^B U''^C \notag\\
& \quad + \bar{\mu}_q q^{ab} \ptl_{U^C} h_{AB} (U) \ptl_a U'^A \ptl_b U^B U'^C \notag\\
& \quad+ \bar{\mu}_q q^{ab} h_{AB}(U) \ptl_a U''^A \ptl_b U^B  
 + \bar{\mu}_q q^{ab}  h_{AB} (U) \ptl_a U'^A \ptl_b U'^B \label{H''}\\
&D^2_{\mbo{\gamma}_\lambda \mbo{\gamma}_s} \cdot H_a (s=0)  \notag\\
&= -4 \leftexp{(q')}{\grad}_b \mbo{\pi}'^b_a -2 \leftexp{(q)}{\grad}_b \mbo{\pi}''^b_a + 2p'_A \ptl_a U'^A + \ptl_a U^A p''_A \label{Ha''} \notag\\
\intertext{where}
\leftexp{(q')}{\grad}_b V^a \fdg =&\, \ptl_b V^a+ \halb q^{ad}  \left ( \leftexp{(q)}{\grad}_b q'_{dc} + \leftexp{(q)}{\grad}_c q'_{bd} -\leftexp{(q)}{\grad}_d q'_{bc}  \right) V^c. 
\end{align}

\noindent We arrive at the following theorem

\begin{corollary}\label{second-var}
Suppose $X'$ is the first variation phase space, then the field equations for the dynamics in $X'$ are given by the variational principle: 
\begin{align}
J_{\textnormal{EWM}} (X'_{\textnormal{EWM}}) \fdg = \int \left( \mbo{\pi}'^{ab} \ptl_t q'_{ab} + p'_A U^{'A} - \halb N H'' - N' H' - N'_a H'_a \right)
\end{align} 
where $H''$, $H'$ $H'_a$ are  \eqref{H''}, $D_{\mbo{\gamma}_s} \cdot H$ and $D \cdot H_a$  at $s=0$ respectively, $N' \fdg = D_{\mbo{\gamma}_s} \cdot N \big\vert_{s=0}$ and $ N'_a \fdg = D_{\mbo{\gamma}_s} \cdot N_a \big\vert_{s=0}.$

\end{corollary}

The approach used above is the classical Jacobian method, as remarked by Moncrief \cite{Moncrief_74}. Separately, it may be noted that the construction of  the wave map field equations is analogous to that of the geodesic deviation equations or the `Jacobi' fields \cite{synge_34, Misn_78}. In view of the fact that the Hamiltonian formulation of the geodesic deviation equations is relatively uncommon, our derivation may also be adapted for this purpose. Finally, we would like to emphasize that our assumption that \eqref{s-t-coordinates} holds, is not (effectively) a restriction in the class of perturbations. In case this condition is relaxed, we shall also pick up the Riemann curvature of the target, but with torsion. The fact that we pick only the curvature term of the target is crucial for our work. We would also like to remark that the deformations which correspond to the coordinate directional derivatives along the curves $\mbo{\gamma}_\lambda$ are equivalent to (induced) covariant deformations on the target, on account of the fact the Kerr wave map is a critical point of \eqref{wm-action-again}.   

The variational principle in Corollary \ref{second-var} and its field equations correspond to a general Weyl-Papapetrou gauge. If we consider further gauge-fixing \eqref{conformal}, where the densitized metric $ \bar{\mu}^{-1}_q q_{ab}$ or equivalently the densitized inverse metric $\bar{\mu}_q q^{ab}$ is fixed, we obtain 

\begin{align}
H' =&   \bar{\mu}_{q_0} (2 \Delta_0 \mbo{\nu}') + \halb \bar{\mu}_{q_0} \ptl_{U^C} h_{AB} q_0^{ab} \ptl_a U^A \ptl_b U^B U'^C \\
&+ \bar{\mu}_{q_0} q_0^{ab} h_{AB} \ptl_a U'^A \ptl_b U^B    \notag\\
H'' =&\bar{\mu}^{-1}_{q_0} ( 2 e^{-2\mbo{\nu}} \Vert \varrho' \Vert^2_{q_0} - \tau'^2 e^{2\nu} \bar{\mu}^2_{q_0} + p'_A p'^A)
+ \bar{\mu}_{q_0} \big( 2 \Delta_0 \mbo{\nu}'' \notag \\
&+ \ptl^2_{U^C U^D} h_{AB} q^{ab}_0 \ptl_a U^A \ptl_b U^B U'^C U'^D + \ptl_C h_{AB} q^{ab}_0 \ptl_a U'^A \ptl_b U^B U'^C \notag\\
&+ \halb \ptl_Ch_{AB} q_{0}^{ab} \ptl_a U^A \ptl_b U^B U''^C + 2 h_{AB} q^{ab}_{0} \ptl_a U''^A \ptl_b U^B \notag\\
&+ 2h_{AB} q^{ab}_0 \ptl_a U'^A \ptl_bU'^B + 2 \ptl_C h_{ab} q^{ab}_0 \ptl_a U'^A \ptl_b U^B U'^C \big) .
\end{align}
The aim of our work is to construct an energy for the linear perturbative theory of Kerr black hole spacetimes, for which the Hamiltonian formulation is naturally suited. In contrast with the Lagrangian variational principles (e.g., \eqref{EH} and \eqref{EWM-Lag}), the Hamiltonian variation principles are not spacetime diffeomorphism invariant. 
In this work, we shall work in the $2+1$ maximal gauge condition. We point out that this gauge condition was also used by Dain-de Austria for the extremal case \cite{DA_14}. We shall need the following statement.

\begin{claim} \label{2+1max}
Suppose $N' \in C^{\infty} (\Sigma),$
\begin{subequations}
\begin{align}
\Delta_0 N' =& \, 0, \quad \textnormal{in the interior of} \quad (\Sigma, q_0), \label{Lap-N}\\
N' \big|_{\ptl \Sigma}=& \,0, \label{Lap-bdry} 
\end{align}
\end{subequations}
then $N' \equiv 0$ on $(\Sigma, q_0).$
\end{claim}
\begin{proof}
If we multiply \eqref{Lap-N} with $N'$ and integrate by parts, we get $\int \vert \grad_0 N \vert^2 =0$ in the interior of $\Sigma,$ after using \eqref{Lap-bdry}. It follows that $N'$ is a constant in $\Sigma.$
\end{proof}

\noindent The variational principle in Corollary \ref{second-var} now gives the following field equations (for smooth and compactly supported variations)
\begin{subequations}
\begin{align}
 h_{AB}(U) \ptl_t U'^B=& \,  e^{2 \mbo{\nu}}\bar{\mu}^{-1}_{q_0} N p'_A + h_{AB}(U) \mathcal{L}_{N'} U^B, \label{u-prime-dot} \\
\ptl_t p'_A =& h_{AB} (U) \leftexp{(h)}{\grad}_a ( N \bar{\mu}_{q} q^{ab} \grad_b  U'^B) \notag \\
& \quad \quad + N \bar{\mu}_q  h_{AE}(U)  R^E_{\,\,\,\,BCD} q^{a b} \ptl_a U^B U'^C \ptl_b U^D \notag\\
=&h_{AB} (U) \leftexp{(h)}{\grad}_a ( N \bar{\mu}_{q_0} q_0^{ab} \grad_b  U'^B) \notag \\
& \quad \quad + N \bar{\mu}_{q_0} h_{AE}(U)  R^E_{\,\,\,\,BCD} q_0^{a b} \ptl_a U^B U'^C \ptl_b U^D  \label{p-prime-dot} \\
\ptl_t q'_{ab}=&   2N  \bar{\mu}^{-1}_{q} \text{CK}_{ab}(Y', q)+ \leftexp{(q)}{\grad} _a N'_b + \leftexp{(q)}{\grad}_b N'_a  \notag\\
=& 2N e^{-2 \nu} \bar{\mu}^{-1}_{q_0} \text{CK}_{ab}(Y', q_0)+ \mathcal{L}_{N'} ( e^{ 2 \mbo{\nu}}(q_0)_{ab}) \\
\ptl_t \mbo{\pi}'^{ab} =&  (\bar{\mu}_q q^{bc} q^{ad})' ( \ptl^2_{dc} N - \leftexp{(q)}{\Gamma} ^f_{cd} \ptl_f N) \notag\\
&+ \bar{\mu}_q q^{bc} q^{ad} ( q^{fl} (\leftexp{(q)}{\grad} _d q'_{l c} + \leftexp{(q)}{\grad} _c q'_{l d}+\leftexp{(q)}{\grad} _l q'_{cd}) \ptl_f N) \notag\\
&+ (\halb N \bar{\mu}_q ( q^{ac} q^{bd} - \halb q^{ab} q^{cd}))' h_{AB}(U) \ptl_a U^a \ptl_b U^B \notag\\
&+ \halb N \bar{\mu}_q (q^{ac} q^{bd} - \halb q^{ab} q^{cd}) \notag\\
& \quad \cdot (2 h_{AB} (U) \ptl_a U'^A \ptl_b U^B + \ptl_{U^C} h_{AB} (U) \ptl_a U^A \ptl_b U^B U'^C)
\end{align}
\end{subequations}
together with the constraints
\begin{subequations}
\begin{align}
H'=& \, 0 \\
\intertext{and}
H'_a=& \, \leftexp{(q)}{\grad}_b \mbo{\pi}'^b_a + p'_A \ptl_a U^A= 0 \label{mom-conf},
\end{align}
in the $2+1$ maximal gauge. 
\end{subequations}
We have, 
\begin{align}
q'= \text{Tr} \, q'_{ab}, \quad \tau' = \bar{\mu}^{-1}_q q_{ab} \mbo{\pi}'^{ab},
\intertext{and}
\ptl_t q' = -2N \tau' + 2 \, \leftexp{(q)}{\grad}^c N'_c,& \quad \ptl_t \tau' = - \Delta_0 N' + N \mathfrak{q}'.
\intertext{In the $2+1$ maximal gauge (cf. Claim \ref{2+1max} )}
\ptl_t q' = 2\leftexp{(q)}{\grad}^c N'_c,& \quad \Delta_0 N' =0.
\end{align}

\noindent 
Let us now formally discuss the structures associated to our Hamiltonian framework. 
 The phase space $X'_{\text{EWM}}$ is such that $(q'_{ab}, U'^A)$ are $C^{\infty}(\Sigma)$ symmetric covariant 2-tensor and smooth vector field respectively and $(\mbo{\pi}'^{ab}, p'_A)$ are $C^{\infty}(\Sigma)$ symmetric 2-tensor densities and scalar density (for each $A$)  respectively, which together form the cotangent bundle $T^*\mathcal{M},$ which we had represented as $X'_{\text{EWM}}.$ The Hamiltonian and momentum constraint spaces $\mathscr{C}_{H'},\mathscr{C}_{H'_a}$ are defined as follows 
 \begin{subequations} \label{constraint-set}
\begin{align}
\mathscr{C}_{H'} =& \Big\{ (q'_{ab}, \mbo{\pi}'^{ab})( U'^A, p'_A) \in T^*\mathcal{M} \,  \big| \, H' =0  \Big\}, \\
\mathscr{C}_{H'_a} =& \Big\{ (q'_{ab}, \mbo{\pi}'^{ab})( U'^A, p'_A) \in T^*\mathcal{M} \, \big| \, H'_a =0,\, a =1,2  \Big\}.
\end{align} 
\end{subequations}
Furthermore, we consider our time coordinate gauge condition to be `$2+1$ maximal': 

\begin{align}
\mathscr{C}_{\tau'} =& \Big\{ (q'_{ab}, \mbo{\pi}'^{ab})( U'^A, p'_A) \in T^*\mathcal{M} \big| \, 
\tau' =0  \Big\}.
\end{align}
In our work we shall be interested in the space 
\begin{align}
\mathscr{C}_{H'} \cap \mathscr{C}_{H'_a} \cap \mathscr{C}_{\tau'}
\end{align}
for our initial value framework. 
In general, proving local existence of Einstein equations using the Hamiltonian initial value problem is a complex problem. 
We note the following statement from the Lagrangian framework of Einstein's equations from the classical result of Choquet-Bruhat and Geroch \cite{Bruhat_Geroch_classic} in $3+1$ dimensions. 

Suppose $\{ (\bar{q}'_{ab}, \bar{\mbo{\pi}}'^{ab})\}_0 \in \mathscr{C}_{\bar{H}'} \cap \mathscr{C}_{\bar{H}'_i},$ then it follows from the classic results of Choquet-Bruhat and Geroch, adapted to our linear perturbation problem, that there exists a unique, regular, maximal development of $\{ (\bar{q}'_{ab}, \bar{\mbo{\pi}}'^{ab}) \}_0,$ $ \iota \fdg \olin{\Sigma} \to \olin{\Sigma} \times \mathbb{R},$ such that  $\{ (\bar{q}'_{ab}, \bar{\mbo{\pi}}'^{ab}) \}_t \in \mathscr{C}_{\bar{H}'} \cap \mathscr{C}_{\bar{H}'_i} $ is causally determined from the initial data $\{ (q'_{ab}, \mbo{\pi}'^{ab}), (U'^A, p'_A) \}_0 $ in a suitable gauge; where $\mathscr{C}_{\bar{H}'}$ and $\mathscr{C}_{\bar{H}'_i}$ are defined analogous to \eqref{constraint-set}. 

\noindent Let us introduce the following notions from the machinery of linearization stability. 
Let us define the constraint map $\olin{\Psi}$ of $(\overline{\Sigma}, \bar{q})$ as a map from the cotangent bundle to a $4-$tuple of scalar densities, $ \olin{\Psi} \fdg T^* \olin{\mathcal{M}} \to \mathcal{C}^{\infty} (\olin{\Sigma}) \times \mathcal{T} \olin{\Sigma},$ such that 
\begin{align}
\olin{\Psi} (\bar{q}, \bar{\mbo{\pi}}) =  (\bar{H}, \bar{H}_i), \quad i = 1, 2, 3.
\end{align}
Let us denote the  deformation of the constraint map as $D \cdot \olin{\Psi} (\bar{q}, \bar{\mbo{\pi}}).$ 
Then the $L^2-$adjoint, $D^\dagger \cdot \olin{\Psi} (\bar{q}, \bar{\mbo{\pi}}) (\bar{C}, \bar{Z}), $ of the deformations  $D \cdot \Psi$  of the constraint map is a
2-tuple (an element of a Banach space) consisting of a covariant symmetric 2-tensor and a contravariant symmetric 2-tensor density and is given by 
\begin{align} \label{3+1adjoint}
D^\dagger \cdot \Psi (\bar{q}, \bar{\mbo{\pi}}) (\bar{C}, \bar{Z}) \fdg =&\, \Big( \bar{\mu}^{-1}_{\bar{q}} \Big(  \halb \Big( \Vert \mbo{\pi} \Vert_{\bar{q}}^2 - \text{Tr}_{\bar{q}}(\bar{\mbo{\pi}} )^2 \Big) \bar{q}^{ij} \bar{C} - 2 \Big( \bar{\mbo{\pi}}^{ik} \bar{\mbo{\pi}}_k^{j} - \halb \mbo{\pi}^{ij} \text{Tr}_{\bar{q}} (\bar{\mbo{\pi}}) \Big) \bar{C} \Big)   \notag\\
&\quad -\bar{\mu}_q \Big( \bar{q}^{ij} \Delta_{\bar{q}} C - \leftexp{(\bar{q})}{\grad}^{i} \leftexp{(\bar{q})}{\grad}^{j} \bar{C} + R^{ij} \bar{C} - \halb \bar{q}^{ij} R_{\bar{q}} \bar{C} \Big) \notag\\
& \quad + \leftexp{(\bar{q})}{\grad}_k (\bar{Z}^k \bar{\mbo{\pi}}^{ij}) - \leftexp{(\bar{q})}{\grad}_k \bar{Z}^i \bar{\mbo{\pi}}^{kj} - \leftexp{(\bar{q})}{\grad}_k \bar{Z}^i \bar{\mbo{\pi}}^{jk}, \notag\\
& \quad   -2 \bar{\mu}_{\bar{q}}^{-1} \bar{C} \Big(\bar{\mbo{\pi}}_{ij} - \halb \text{Tr} (\bar{\mbo{\pi}}) \bar{q}_{ij} \Big) - \leftexp{(q)}{\grad}_i \bar{Z}_j - \leftexp{(q)}{\grad}_j \bar{Z}_i \Big)
\end{align}
The expression \eqref{3+1adjoint} is closely related to the $L^2-$adjoint of the Lichnerowicz operator. 
Moncrief had characterized the splitting theorem, established by Fischer-Marsden \cite{FM_75}, of the (Banach) spaces acted on by the constrant map 
\begin{align}
 \text{ker}\,D^\dagger \cdot \Psi  (\bar{q}, \bar{\mbo{\pi}} ) (\bar{C}, \bar{Z})  \,\oplus\, \text{range} \, D \cdot \Psi  (\bar{q}, \bar{\mbo{\pi}} )(\bar{q}', \bar{\mbo{\pi}}' ),
\end{align}
by associating the kernel of the adjoint operator ($ \text{ker}\, D^\dagger \cdot \olin{\Psi} (\bar{q}, \bar{\mbo{\pi}} ) (\bar{C}, \bar{Z})$) to the existence of spacetime Killing isometries. In particular, Moncrief proved that $\text{ker} \, D^\dagger \cdot \olin{\Psi} (\bar{C}, \bar{Z})$  is non-empty if and only if there exists a spacetime Killing vector. 
 This result is crucial for our work, but in the dimensionally reduced framework. In the following, we shall establish equivalent results in our dimensionally reduced perturbation problem.
\begin{lemma}
Suppose  $(M, g)$ is the $2+1$ spacetime obtained from the dimensional reduction of the axially symmetric, Ricci-flat $3+1$ spacetime $(\bar{M}, \bar{g})$ and $D \cdot \Psi$ is the deformation around the Kerr metric of the constraint map $\Psi$ of the dimensionally reduced $2+1$ Einstein-wave map system on $(M, g)$ then 

\begin{enumerate}
\item The adjoint $D^\dagger \cdot \Psi (q', \mbo{\pi}') (C, Z) $ of the constraint map $\Psi$ is given by 

\begin{align} \label{2+1adjoint}
D^{\dagger} \cdot \Psi =& \Big( \bar{\mu}_q (\leftexp{(q)}{\grad}^b \,\leftexp{(q)}{\grad}^a C -q^{ab}\, \leftexp{(q)}{\grad}_c \leftexp{(q)}{\grad}^c C ) \notag\\
&\quad + \halb C \bar{\mu}_q h_{AB} (q^{ac}q^{bd} - \halb q^{ab} q^{cd})\ptl_a U^A \ptl_b U^B,  \notag\\
&  \quad - \leftexp{(q)}{\grad}_a Z_b - \leftexp{(q)}{\grad}_b Z_a \Big)
\end{align}

\item The kernel $(\text{ker} (D^\dagger \Psi))$ of the adjoint of the constraint map $\Psi$  is one dimensional and is equal to
$(N, 0)^T$
\end{enumerate}  
\end{lemma}
\begin{proof}
Consider the constraint map $\Psi$
\begin{align}
\Psi (g, \mbo{\pi}) =& (H , H_i)
\end{align}
then from the deformation of $\Psi, D \cdot \Psi (q, \mbo{\pi}) (q', \mbo{\pi}') \fdg = (H', H'_a),$ around a general metric,
it follows that its $2+1$ $L^2-$adjoint is given by

\begin{align} \label{2+1adjoint-general}
D^{\dagger} \cdot \Psi =& \Big( \halb C \bar{\mu}^{-1}_q q^{ab} (\Vert \mbo{\pi} \Vert^2_q - \text{Tr}(\mbo{\pi})^2)- 2C \bar{\mu}^{-1}_q \left(  \mbo{\pi}^{ac} \mbo{\pi}^{b}_c - \mbo{\pi}^{ab} \textnormal{Tr}_q(\mbo{\pi}) \right) \notag\\
&+\bar{\mu}_q (\leftexp{(q)}{\grad}^b \,\leftexp{(q)}{\grad}^a C -q^{ab}\, \leftexp{(q)}{\grad}_c \leftexp{(q)}{\grad}^c C ) \notag\\
&+ \leftexp{(q)}{\grad}_c(\mbo{\pi}^{ab} Z^c) - \leftexp{(q)}{\grad}_c Z^a \mbo{\pi}^{cb} - \leftexp{(q)}{\grad}_c Z^b \mbo{\pi}^{ca}\notag\\
&+ \frac{1}{4}\bar{\mu}_q^{-1} C q^{ab} p_A p^A + \halb C \bar{\mu}_q h_{AB} (q^{ac}q^{bd} - \halb q^{ab} q^{cd})\ptl_a U^A \ptl_b U^B,  \notag\\
& -2C \bar{\mu}^{-1}_q (\mbo{\pi}_{ab} - q_{ab} \textnormal{Tr}_q \mbo{\pi}) - \leftexp{(q)}{\grad}_a Z_b - \leftexp{(q)}{\grad}_b Z_a \Big),
\end{align}

\noindent analogous to \eqref{3+1adjoint}, while noting that the (dimensionally reduced) wave map variables are not constrained due to the introduction of the twist potential, after using the Poincar\'e Lemma (see e.g., \eqref{2+1constraints} and then \eqref{2+1-no-constraints}). The expression \eqref{2+1adjoint} follows for the case of dimensionally reduced Kerr metric. 
Now assume that $(C, Z) \in \text{ker}\, D^\dagger \cdot \Psi.$ It follows from \eqref{2+1adjoint} that a vector $K = K_{\perp} \mbo{n} + K_{\parallel} $
satisfies 
\begin{align}\label{Killing}
\leftexp{(g)}{\grad}_\a K_\b + \leftexp{(g)}{\grad}_\b K_\a =0
\end{align} 
with $K_{\perp} = C$ and $K_{\parallel} = Z$
which implies that $K = (C, Z)$ is a (spacetime) Killing vector in $(M, g).$ Conversely, assuming that \eqref{Killing} holds
it follows that the LHS of \eqref{2+1adjoint} vanishes, which implies $K \in \text{ker}\, D^\dagger \cdot \Psi.$ In particular, for the dimensionally reduced Kerr metric $(M ,g)$ the only remaining linearly independent Killing vector is $\ptl_t,$ so $(C, Z)^{\text{T}} \equiv (N, 0)^{\text{T}},$ which, as will be shown later, resolves \textbf{(P2)}.
\end{proof}

In the following, we shall establish that the Hamilton vector field $(H', H'_a)$ is tangential to the flow of the phase space variables $\{ (q', \pi'), (U'^A, p'_A) \}$ in $\mathscr{C}_{H'} \cap \mathscr{C}_{H'_a} \cap \mathscr{C}_{\tau'}.$ 

\begin{lemma}
Suppose $H'$ and $H'_a$ are the linearized Hamiltonian and momentum constraints of the $2+1$ Einstein-wave map system, then their propagation equations are
\begin{subequations} 
\begin{align}
\frac{\ptl}{\ptl t} H' =& q^{ab}  \ptl_a N H'_b  + \ptl_b (N q^{ab} H'_a) \label{h-dot}\\
\frac{\ptl}{\ptl t} H'_a=& \ptl_a N H' \label{ha-dot}
\intertext{and}
N H' =& \ptl_b ( N \bar{\mu}_q q^{ab} h_{AB}  U'^{A} \ptl_b U^B - 2 \bar{\mu}_q q^{ab} \ptl_a N \mbo{\nu}') \label{LS}
\end{align}
\end{subequations}
\end{lemma}
\begin{proof}
The statements \eqref{h-dot} and \eqref{ha-dot}  follow from the linearized  and background (exact) field equations of our $2+1$ Einstein-wave map system. For simplicity in computations, we shall perform our computations  with $q'_0$
held fixed. 
Recall, 
\begin{align}
\ptl_t \varrho'^a_b =& N \bar{\mu}_{q_0} ( q^{ac}_0 \delta^d_b - \halb q_0^{cd} \delta^a_b) (h_{AB} \ptl_c U'^A \ptl_d U^B + \halb \ptl_{U^C} h_{AB} (U) \ptl_c U^A \ptl_d U^B U'^C ) \notag\\
&+ \bar{\mu}_{q_0} q_0^{cd} \delta^a_b \ptl_c N \ptl_d \mbo{\nu}' - \bar{\mu}_{q_0} q^{ac}_0 ( \ptl_b N \ptl_c \mbo{\nu}') \\
\ptl_t \mbo{\nu}'=& \frac{1}{2 \bar{\mu}_{q_0}} \ptl_c (\bar{\mu}_{q_0} N'^c) + 2 \mathcal{L}_{N'} \mbo{\nu}
\end{align}
Consider the quantities, 
\begin{align}
\ptl_t \big(2 \bar{\mu}^{-1}_{q_0} \ptl_b ( \bar{\mu}_{q_0} q^{ab}_0 \ptl_a \mbo{\nu}') \big) =& \bar{\mu}^{-1}_{q_0} \ptl_b ( \bar{\mu}_{q_0} q^{ab}_0 \ptl_a ( \bar{\mu}^{-1}_{q_0} \ptl_c (\bar{\mu}_{q_0} N'^c) + 2 N'^c\ptl_c \mbo{\nu}))
\end{align}
\begin{align}
\halb \bar{\mu}_q q^{ab}  \ptl_{U^C} h_{AB} \ptl_a U^A \ptl_b U^B  \ptl_t U'^C =&  \halb N q^{ab}  \ptl_{U^C} h_{AB} \ptl_a U^A \ptl_b U^B   p'^C \notag\\
&+ \halb \bar{\mu}_q q^{ab}  \ptl_{U^C} h_{AB} \ptl_a U^A \ptl_b U^B \mathcal{L}_{N'} U^C  \\
  \bar{\mu}_q q^{ab} h_{AB}(U) \ptl_a (\ptl_t U'^A) \ptl_b U^B=& \bar{\mu}_q q ^{ab} h_{AB}(U) \ptl_a ( \bar{\mu}^{-1}_q N p'^A)
  \ptl_b U^B \notag\\
  &+ \bar{\mu}_q q ^{ab} h_{AB}(U) \ptl_a  (\mathcal{L}_{N'} U^A) \ptl_b U^B \notag\\
  =& \bar{\mu}_q q ^{ab} h_{AB}(U) \ptl_a ( \bar{\mu}^{-1}_q N p'^A)
  \ptl_b U^B \notag\\
  &+ \bar{\mu}_q q ^{ab} h_{AB}(U) \mathcal{L}_{N'} (\ptl_a U^A) \ptl_b U^B
\end{align}
Combining the results above and noting that for our gauge,
\begin{align}
\text{CK}^{ab} (N', q_0) = \bar{\mu}_{q_0} (\leftexp{(q_0)}{\grad}^a N'^b + \leftexp{(q_0)}{\grad}^b N'^a - q^{ab}_0 \leftexp{(q_0)}{\grad}_c N'^c ) = - 2 N e^{-2 \nu} q^{bc}_0 \varrho'^a_c
\end{align}
we get 
\begin{align}
\ptl_t H' =& q^{ab} \ptl_b N (-2 \leftexp{(q_0)}{\grad}_c \varrho'^c_a + p'_A \ptl_a U^A) -2 \ptl_b ( N q^{ab} \leftexp{(q_0)}{\grad}_c \varrho'^c_a) + \ptl_b (N q^{ab} p'_A \ptl_a U^A)  \\
=&  q^{ab}  \ptl_a N H'_b  + \ptl_b (N q^{ab} H'_a)
\end{align}
Likewise, for \eqref{ha-dot}, consider, 

\begin{align}
 \leftexp{(q_0)}{\grad}_a (\ptl_t \varrho'^a_b) =& \leftexp{(q_0)}{\grad}_a \big(N \bar{\mu}_{q_0} ( q^{ac}_0 \delta^d_b - \halb q_0^{cd} \delta^a_b) (h_{AB} \ptl_c U'^A \ptl_d U^B 
 \notag\\
 &+ \halb \ptl_{U^C} h_{AB} (U) \ptl_c U^A \ptl_d U^B U'^C ) \notag\\
 &+\bar{\mu}_{q_0} q_0^{cd} \delta^a_b \ptl_c N \ptl_d \mbo{\nu}' - \bar{\mu}_{q_0} q^{ac}_0 ( \ptl_b N \ptl_c \mbo{\nu}') \big)  \\
\ptl_a U^A \ptl_t p'_A =& h_{AB} \ptl_a U^A \leftexp{(h)}{\grad}_c ( N \bar{\mu}_q q^{cb} \leftexp{(h)}{\grad}_b U'^B) \notag\\
&+ N \bar{\mu}_q h_{AB} (U) \ptl_a U'^B R^{E}_{\,\,\,\,BCD}   q^{ab} \ptl_a U^B \ptl_b U^D U'^C 
\end{align}

Now combining all the above, we have
\begin{align}
\ptl_t H'_c = \ptl_c N H'
\end{align}
in view of the background field equations \eqref{Kerr-p-dot} and \eqref{kerr-maximal}. 
For \eqref{LS}, first note that 

\begin{align} \label{var-from-chris}
 N \bar{\mu}_{q_0} \ptl_{U^C} h_{AB} (U) q^{ab}_0 \ptl_a U^A \ptl_b U^C U'^B =& N \bar{\mu}_{q_0} h_{AB}\leftexp{(h)}{\Gamma}^A_{CD} (U) q^{ab}_0 \ptl_a U^C \ptl_b U^D U'^B \notag\\
 &- \halb N \bar{\mu}_{q_0} \ptl_{U^C} h_{AB}(U) q^{ab}_0 \ptl_a U^A \ptl_b U^B U'^C
\end{align}
after a suitable relabelling of the indices. Now consider 

\begin{align} \label{LS-computation}
N H' =& 2N \bar{\mu}_{q_0} \Delta_0 \mbo{\nu}' +  N \bar{\mu}_q h_{AB}\leftexp{(h)}{\Gamma}^A_{CD} (U) q^{ab}_0 \ptl_a U^C \ptl_b U^D U'^B \notag\\
&- N \bar{\mu}_{q_0} \ptl_{U^C} h_{AB} (U) q^{ab}_0 \ptl_a U^A \ptl_b U^C U'^B + \bar{\mu}_{q_0} q^{ab}_0 h_{AB} \ptl_a U'^A \ptl_b U^B \notag\\
=&  2N \bar{\mu}_{q_0} \Delta_0 \mbo{\nu}'+ \ptl_b (N \bar{\mu}_{q_0} q_0^{ab}h_{AB} \ptl_a U^A U'^B) \notag\\
=& \ptl_b ( -2 \bar{\mu}_{q_0} q_0^{ab} \ptl_a N \mbo{\nu}' + 2 \ptl_a \mbo{\nu}' + N \bar{\mu}_q q^{ab} h_{AB} \ptl_a U^A U'^B )
\end{align}
where we have used \eqref{var-from-chris} and the background field equations \eqref{Kerr-p-dot} and \eqref{kerr-maximal}. 
Fundamentally, underlying the statement \eqref{LS-computation} is the fact that $(N, 0)^{\text{T}}$ is the kernel of the adjoint of the constraint map of our linear perturbation theory.  
\end{proof}


\section{A Positive-Definite Hamiltonian Energy from Negative Curvature of the Target and the Hamiltonian Dynamics}

In arriving at the variational principles and their corresponding field equations, we have used smooth compactly supported deformations. In the construction of a Hamiltonian energy function the underlying computations are bit more subtle, in connection with the boundary terms and the initial value problem. We impose the regularity conditions on the axis of initial hypersurface $\Sigma$ by fiat, so that the fields smoothly lift up to the original $\olin{\Sigma}.$  We shall assume the following conditions on the two disjoint segments of the axes $\Gamma = \Gamma_1 \cup \Gamma_2$. In the wave map $U \fdg M \to \mathbb{N},$ one of the components corresponds to the norm of the Killing vector $\vert \Phi \vert$ and the other the `twist'. For the twist component we assume
\begin{align}
U'^A\vert_{\Gamma_1} = U'^A \vert_{\Gamma_2}, \quad \text{for the corresponding A}
\end{align}
on account of our assumption that the perturbation of the angular momentum is zero. Without (effective) loss of generality we assume,
\begin{align}
U'^A = 0 \quad \text{on} \quad \Gamma \quad \text{which implies} \quad \ptl_{\vec{t}} U'^A=0 
\end{align}
where $\ptl_{\vec{t}}$ is the derivative tangent to the axis. To prevent conical singularity on the axis, which, as we remarked, allows us to smoothly lift our fields up to the original manifold $\olin{\Sigma},$ we assume

\begin{align}
\vert \Phi \vert' =0, \quad \ptl_{\vec{t}} \vert \Phi \vert' =0
\end{align} 
for the `norm' component of $U$; and
\begin{align}
\ptl_{\vec{n}} U'^A =0, \quad p'_A =\ptl_{\vec{t}} p'_A =0, \quad \ptl_{\vec{n}} p'_A =0,
\end{align} 
where $\ptl_{\vec{n}}$ is the derivative normal to the axes $\Gamma.$
In this work, for $\mbo{\nu}'$ we shall assume 
\begin{align}
\ptl_{\vec{n}} \mbo{\nu}' =0
\end{align}
which corresponds to the preservation of the condition that inner (horizon) boundary is a minimal surface. Now define an `alternative' Hamiltonian constraint $H'^{\text{Alt}}$

\begin{align}
H'^{\text{Alt}} \fdg = \bar{\mu}_{q_0} (2 \Delta_0 \mbo{\nu}'  + h_{AB}(U) U'^B (\Delta_0 U^A + \leftexp{(h)}{\Gamma}^A_{BC} q^{ab}_0 \ptl_a  U^B \ptl_b U^C))
\end{align}
where we have now used the following identity to transform from $H':$

\begin{align}
&\halb \ptl_{U^C} h_{AB} q^{ab} \ptl_a U^A \ptl_b U^B U'^C + h_{AB}(U) q^{ab} \ptl_a U'^A \ptl_b U^B \notag\\
&= h_{AB} U'^B (\Delta_q U'^A + \leftexp{(h)}{\Gamma}^A _{CD} q^{ab} \ptl_a U^C \ptl_b U^D ) 
\end{align}
which is analogous to \eqref{christof-intro}, but now for the $q$ metric. 
Analogously define

\begin{align}
H''^{\text{Alt}} \fdg =& \bar{\mu}^{-1}_{q_0} (2 e^{-2 \mbo{\nu}} \Vert \varrho' \Vert^2_{q_0} - \tau'^2 e^{2 \mbo{\nu}} \bar{\mu}^2_{q_0} + p'_A p'^A) \notag\\
&- \bar{\mu}_{q_0} h_{AB} U'^B( \leftexp{(h)}{\Delta} U'^A +R^A_{BCD} q_0^{ab} \ptl_a U^B \ptl_b U^C U'^D)
\end{align}
Using a further divergence identity: 
\begin{align}
& \leftexp{(h)}{\grad}_a (h_{AB} U'^B \leftexp{(h)}{\grad} ^a U'^A) - h_{AB} U'^B  \leftexp{(h)}{\grad}_a  \leftexp{(h)}{\grad}^a U'^A \notag\\ 
&= h_{AB} q^{ab} \leftexp{(h)}{\grad}_a U'^A \leftexp{(h)}{\grad}_b U'^B , \quad (\Sigma, q),
\end{align}
let us now define our `regularized' Hamiltonian energy density as

\begin{align} \label{e-reg-den}
\mathbf{e}^{\text{Reg}} \fdg=& N \bar{\mu}^{-1}_{q_0} e^{-2 \mbo{\nu}} \left( \Vert \varrho' \Vert^2_{q_0} + \halb p'_A p'^A  \right) - \halb N e^{2 \mbo{\nu}} \bar{\mu}_{q_0} \tau'^2 \notag\\
&+ \halb N \bar{\mu}_{q_0} q_0^{ab} h_{AB}(U) \leftexp{(h)}{\grad}_a U'^A \leftexp{(h)}{\grad}_b U'^B \notag\\
&- \halb N \bar{\mu}_{q_0} q^{ab}_0 h_{AE} (U) U'^A \leftexp{(h)}R^E _{BCD} \ptl_a U^B \ptl_b U^C U'^D 
\end{align} 
and the `regularized' Hamiltonian $H^{\text{Reg}}$
\begin{align} \label{e-reg}
H^{\text{Reg}} \fdg = \int_{\Sigma} \mathbf{e}^{\text{Reg}} \, d^2 x.
\end{align}
It is evident that $H^{\text{Reg}}$ is manifestly positive-definite in the maximal gauge $\tau' \equiv 0,$ in view of the fact that the target is the (negatively curved)   hyperbolic 2-plane. Indeed, we obtain a similar energy expression \eqref{e-reg-den} and \eqref{e-reg} in the higher dimensional case where the target for wave maps is $SL(n-2)/SO(n-2).$
As we already alluded to, the purpose of distinguishing the quantities $H^{\text{Reg}}$ is that they are transformed, using divergence identities, from $H$ , and thus differ by boundary  terms. In case the perturbations are compactly supported in $(\Sigma)$ it is immediate that they have the same value. In general, dealing with all the boundary terms and their evolution in our problem is considerably subtle \cite{GM17}. 
  
It is conjectured that our Hamiltonian energy functional and the boundary terms constitute deformations of the ADM mass of $(\olin{\Sigma}, \bar{q})$ at the outer boundary. These technical aspects in our problem, related to the boundary behaviour in the quotient space, shall be completed systematically in a separate work.  The aim of this work is the construction of the positive-definite energy functional $H^{\text{Reg}}.$ In the following, we shall establish the validity and consistency of our approach to construct the energy using two separate methods. Firstly, we shall show that $H^{\text{Reg}}$ serves as a Hamiltonian that drives the dynamics of the unconstrained or `independent' phase-space variables. Secondly, we shall establish that there exists a spacetime divergence-free vector density, whose flux through $\Sigma$ is $H^{\text{Reg}}.$

We would like to point out that, in our problem, the $2+1$ geometric phase space variables (e.g., $\mbo{\nu}', \varrho'^a_b$) are completely determined by the constraints and gauge-conditions. Therefore, their Hamiltonian dynamics are governed by the `independent' or `unconstrained' dynamical variables $(U'^A, p'_A)$. In the following, we shall prove that $H^{\text{Reg}}$ drives the \emph{coupled} Einstein-wave map dynamics of $(U'^A, p'_A)$. The proof that the equivalent $H^{\text{Reg}}$ serves as the Hamiltonian for the reduced Einstein-Maxwell phase space can be found in Section 5 of \cite{GM17}.

\begin{theorem}
Suppose the globally regular, maximal development of the smooth, compactly supported perturbation initial data in the domain of outer communications of the Kerr metric is such that $\big \{ (q'_{ab}, \mbo{\pi}'^{ab}), (U'^A, p'_A) \big \}_{t} \in \mathscr{C}_{H'} \cap \mathscr{C}_{H'_a} \cap \mathscr{C}_{\tau'},$ then functional $H^{\text{Reg}}$ is a Hamiltonian for the coupled dynamics of $(U'^A, p'_A):$
\begin{subequations}
\begin{align}
D_{p'_A} \cdot H^{\textnormal{Reg}} =&\, \ptl_t U'^A \\
D_{U'^{A}} \cdot H^{\textnormal{Reg}} =&\,- \ptl_t p'_A.
\end{align}
\end{subequations}
\end{theorem}
\begin{proof}
The first variation $D_{p'_A} \cdot H^{\text{Reg}}$ contains the terms:
\begin{align}
\varrho'^a_b \varrho''^a_b =& \halb N^{-1} \bar{\mu}_q (\leftexp{(q)}{\grad}^b N'_a + \leftexp{(q)}{\grad}_a N'^b - \delta_a^b \leftexp{(q)}{\grad}_c N'^c)  \notag\\
& ( \leftexp{(q)}{\grad}^a N''_b + \leftexp{(q)}{\grad}_b N''^a - \delta_b^a \leftexp{(q)}{\grad}_c N''^c) \notag\\
&= \halb N^{-1} \bar{\mu}_q \leftexp{(q)}{\grad}_a N'^b  ( \leftexp{(q)}{\grad}^a N''_b + \leftexp{(q)}{\grad}_b N''^a - \delta_b^a \leftexp{(q)}{\grad}_c N''^c)
\end{align}
We have the divergence identity:
\begin{align}
& \leftexp{(q)}{\grad}_a (N^{-1} N'^b \bar{\mu}_q  ( \leftexp{(q)}{\grad}^a N''_b + \leftexp{(q)}{\grad}_b N''^a - \delta_b^a \leftexp{(q)}{\grad}_c N''^c) ) \notag\\
& N^{-1} N'^b ( \leftexp{(q)}{\grad}_a ( \leftexp{(q)}{\grad}^a N''_b + \leftexp{(q)}{\grad}_b N''^a - \delta_b^a \leftexp{(q)}{\grad}_c N''^c))  \notag\\
&+ N^{-1} \bar{\mu}_q \leftexp{(q)}{\grad}_a N'^b  ( \leftexp{(q)}{\grad}^a N''_b + \leftexp{(q)}{\grad}_b N''^a - \delta_b^a \leftexp{(q)}{\grad}_c N''^c) \notag\\
&= -N'^b (\ptl_b U^A p''_A) +  N^{-1} \bar{\mu}_q \leftexp{(q)}{\grad}_a N'^b  ( \leftexp{(q)}{\grad}^a N''_b + \leftexp{(q)}{\grad}_b N''^a - \delta_b^a \leftexp{(q)}{\grad}_c N''^c)
\end{align}
after using the momentum constraint; and 
\begin{align}
D_{p'_A} \cdot \halb p'_A p'^A = p'^A.
\end{align}
Collecting the terms above, gives the Hamilton equation 
\begin{align}
D_{p'_A} \cdot H^{\text{Reg}} =& \, N \bar{\mu}^{-1}_q p'^A + N'^b \ptl_b U^A  \notag\\
=& \,  \ptl_t U'^A.
\end{align}
The quantity $D_{U'^A} \cdot H^{\text{Reg}}$ contains the terms: 

\begin{align}
D_{U'^A} \cdot \halb  h_{AB}(U) \leftexp{(h)}{\grad}_a U'^A \leftexp{(h)}{\grad}_b U'^B =& N \bar{\mu}_q q^{ab} h_{AB} \leftexp{(q)}{\grad}_a U''^A \leftexp{(q)}{\grad}_b U'^B \notag\\
\intertext{note that}
\leftexp{(q)}{\grad}_a (N \bar{\mu}_q q^{ab} h_{AB} U''^A \ptl_b U'^B)=& U''^A \leftexp{(q)}{\grad}_a (N \bar{\mu}_q q^{ab} h_{AB} \ptl_a U'^B) \notag\\
&+ N \bar{\mu}_q q^{ab} h_{AB} \leftexp{(q)}{\grad}_a U''^A \leftexp{(q)}{\grad}_b U'^B
\end{align}
and
\begin{align}
&D_{U'^A} \cdot (- \halb N \bar{\mu}_{q} q^{ab} h_{AE} (U) U'^A \leftexp{(h)}R^E _{BCD} \ptl_a U^B \ptl_b U^C U'^D ) \notag\\
=& \, -  N \bar{\mu}_q q^{ab} h_{AE} (U) \leftexp{(h)}R^E _{BCD} \ptl_a U^B \ptl_b U^C U'^D 
\end{align}
which combine to give 
\begin{align}
D_{U'^A} \cdot H^{\text{Reg}} =& - \leftexp{(q)}{\grad}_a (N \bar{\mu}_q q^{ab} h_{AB} \ptl_a U'^B) \notag\\
-&  N \bar{\mu}_q q^{ab} h_{AE} (U) \leftexp{(h)}R^E _{BCD} \ptl_a U^B \ptl_b U^C U'^D 
\intertext{which is the Hamilton equation}
=& -\ptl_t p'_A.
\end{align} 
\end{proof}
\begin{theorem}
Suppose the variables $\{ (q'_{ab}, \mbo{\pi}'^{ab}), (U'^A, p'_A) \} \in \mathscr{C}_{H'} \cap \mathscr{C}_{H'_a}  \cap \mathscr{C}_{\tau'},$ then there exists a (spacetime) divergence-free vector field density such that its flux through $t-$constant hypersurfaces is  $H^{\text{Reg}}$ (positive-definite). 
\end{theorem}
\begin{proof}
In the proof we shall use the perturbation evolution equations and the background (Kerr metric) field equations.  Consider $\ptl_t \mathbf{e}^{\text{Reg}}$ and it contains the following terms: 
\begin{enumerate}
\item \begin{align} \label{term1}
N \bar{\mu}^{-1}_q p'^A \ptl_t p'_A =&  N \bar{\mu}^{-1}_q p'^A \big (h_{AB} \leftexp{(h)}{\grad}_a (N \bar{\mu}_q q^{ab} \leftexp{(h)}{\grad}_b U'^B) \notag\\
& + N \bar{\mu}_q h_{AB} R^E_{\,\,\,\, BCD} q^{ab} \ptl_a U^B \ptl_b U^D U'^C \big),
\end{align}
\item \begin{align} \label{term2}
 &N \bar{\mu}_q q^{ab} \leftexp{(h)}{\grad}_a (\ptl_t U'^A) \leftexp{(h)}{\grad}_b U'^B \notag\\
 &= N \bar{\mu}_q q^{ab} h_{AB}(U) \leftexp{(h)}{\grad}_b U'^B (\leftexp{(h)}{\grad}_a (\bar{\mu}^{-1}_q Np'^A + \mathcal{L}_{N'} U^A)) 
 \end{align}
 Note the divergence relation involving the terms from \eqref{term1} and \eqref{term2}. 
 \begin{align}
 \leftexp{(h)}{\grad}_a (N^2 q^{ab} h_{AB} p'^A \leftexp{(h)}{\grad}_b U'^B) =& \bar{\mu}^{-1}_q N p'^A \leftexp{(h)}{\grad}_a (N \bar{\mu}_q q^{ab} \leftexp{(h)}{\grad}_b U'^B) \notag\\
 & \quad+ N \bar{\mu}_q q^{ab} h_{AB} \leftexp{(h)}{\grad}_b U'^B \leftexp{(h)}{\grad}_a (N \bar{\mu}^{-1}_q p'^A)
 \end{align}
\item  \begin{align}
  &N \bar{\mu}_q h_{AE}(U) \ptl_t U'^A R^A_{\,\,\,\,BCD} q^{ab} \ptl_a U^B \ptl_b U^D U'^C  \notag\\
  &=N \bar{\mu}_q h_{AE} (\bar{\mu}^{-1}_q N p'_A + \mathcal{L}_{N'} U^A) R^E_{\,\,\,BCD} q^{ab} \ptl_a U^B \ptl_b U^D U'^C
\end{align}

\item 
\begin{align}
& e^{-2 \mbo{\nu}} N \bar{\mu}^{-1}_{q_0} \varrho'^c_a \ptl_t \varrho'^a_c  \notag\\
&= N (\leftexp{(q)}{\grad}_a N'^b + \leftexp{(q)}{\grad}_a N'^b - \delta^b_a \leftexp{(q)}{\grad}_c N'^c) \Big(  N \bar{\mu}_q  (q^{ac} \delta^d_b - \halb \delta^a_b q^{cd}) \notag\\ 
& \quad \quad \cdot ( h_{AB} \ptl_b U'^A \ptl_d U^B + \halb  \ptl_{U}h_{AB} \ptl_b U'^A \ptl_d U^B  )  +  2 \bar{\mu}_q q^{ac}\ptl_c N \ptl_b \mbo{\nu}' \notag\\
& \quad \quad \quad \quad - \bar{\mu}_q \delta^a_b q^{cd} \ptl_c N \ptl_d \mbo{\nu}'\Big) \\
&= \mathcal{L}_{N'} (\bar{\mu}_{q_0} q^{ab}_0) (h_{AB} \ptl_a U'^A \ptl_b U^B + \halb  \ptl_{U^C} h_{AB} \ptl_a U^A \ptl_b U^B U'^C -2 \ptl_a N \ptl_b \mbo{\nu}' ).
\end{align}

\end{enumerate}
Consider the following divergence identities: 
\begin{subequations} \label{div-ident-mom}
\begin{align}
\leftexp{(q_0)}{\grad}^a (N'^b  \ptl_a N \ptl_b \mbo{\nu}' \bar{\mu}_{q_0}) =& \bar{\mu}_{q_0}  \leftexp{(q_0)}{\grad}^a N'^b  \ptl_a N \ptl_b \mbo{\nu}' \notag\\
& \quad + N'^b \leftexp{(q_0)}{\grad}^a ( \bar{\mu}_{q_0} \ptl_a N \ptl_b \mbo{\nu}') \\
\leftexp{(q_0)}{\grad}^b ( N'^a \ptl_a N \ptl_b \mbo{\nu}' \bar{\mu}_{q_0})=& \bar{\mu}_{q_0}  \leftexp{(q_0)}{\grad}^b N'^a  \ptl_a N \ptl_b \mbo{\nu}'  \notag\\
& \quad+ N'^a \leftexp{(q_0)}{\grad}^b (\bar{\mu}_{q_0} \ptl_a N \ptl_b \mbo{\nu}') \\
\leftexp{(q_0)}{\grad}_c ( N'^c q^{ab}_0 \ptl_a N \ptl_b \mbo{\nu}' \bar{\mu}_{q_0}) =& \leftexp{(q_0)}{\grad}_c N'^c( q^{ab}_0 \ptl_a N \ptl_b \mbo{\nu}' \bar{\mu}_{q_0} ) \notag\\
&\quad + N'^c \leftexp{(q_0)}{\grad}_c ( q^{ab}_0 \ptl_a N \ptl_b \mbo{\nu}' \bar{\mu}_{q_0}).
\end{align}
\end{subequations}
Based on the right-hand sides of the divergence identities in \eqref{div-ident-mom} we get, after using the background field equation \eqref{kerr-maximal},

\begin{align}
&-2  \mathcal{L}_{N'} (\bar{\mu}_q q^{ab})  \ptl_a N \ptl_b \mbo{\nu}' \notag\\
&= -2 N'^b \ptl_b N \ptl_a (\bar{\mu}_q q^{ab} \ptl_a \mbo{\nu}') + 2 \leftexp{(q_0)}{\grad}_a ( N'^c \ptl_c \mbo{\nu}' \bar{\mu}_{q_0} q^{ab}_0 \ptl_b N ) \notag\\
& \quad + 2 \leftexp{(q_0)}{\grad}_a ( N'^c \ptl_c N \bar{\mu}_{q_0} q^{ab}_0 \ptl_b \mbo{\nu}' ) - 2 \leftexp{(q_0)}{\grad}_c ( N'^c \ptl_a \mbo{\nu}' \bar{\mu}_{q_0} q^{ab}_0 \ptl_b N ) \notag\\
&= \mathcal{L}_{N'} N (- H' + h_{AB} \ptl_a U'^A \ptl_b U^B + \halb  \ptl_{U^C} h_{AB} \ptl_a U^A \ptl_b U^B U'^C ) \notag\\
& \quad + 2 \leftexp{(q_0)}{\grad}_b (\mathcal{L}_{N'} \mbo{\nu}' \bar{\mu}_{q_0} \ptl^b N+ \mathcal{L}_{N'} N \bar{\mu}_{q_0} \ptl^b \mbo{\nu'} - \bar{\mu}_{q_0} N'^b \ptl_a N \ptl^a \mbo{\nu}').
\end{align}
Now let us focus on the remaining `shift' terms: 
\begin{align}
&N \bar{\mu}_q q^{ab} h_{AB} \leftexp{(h)}{\grad}_b U'^B \leftexp{(h)}{\grad}_a (\mathcal{L}_{N'} U^A), \label{grad-square}\\ 
&- N \bar{\mu}_q h_{AE} ( \mathcal{L}_{N'} U^A) R^E_{\,\,\,BCD} q^{ab} \ptl_a U^B \ptl_b U^D U'^C \\
\intertext{and}
& \mathcal{L}_{N'} (\bar{\mu}_{q_0} q^{ab}_0) (h_{AB} \ptl_a U'^A \ptl_b U^B + \halb  \ptl_{U}h_{AB} \ptl_a U'^A \ptl_b U^B)
\end{align}
Consider the quantity $N \bar{\mu}_q q^{ab} h_{AB} \ptl_a U'^A (\mathcal{L}_{N'}(\ptl_b U^B) )$ that occurs in \eqref{grad-square}, we have 
\begin{align}
&N \bar{\mu}_q q^{ab} h_{AB} \ptl_a U'^A (\mathcal{L}_{N'}(\ptl_b U^B)) + \ptl_{U^C} h_{AB} \mathcal{L}_{N'} U^C N \bar{\mu}_qq^{ab}\ptl_a U'^A \ptl_b U^B    \notag\\
=& \ptl_a U'^A \mathcal{L}_{N'} ( N \bar{\mu}_q q^{ab} h_{AB}(U) \ptl_b U^B) \notag\\
& \quad -\mathcal{L}_{N'} N (h_{AB} \bar{\mu}_q q^{ab} \ptl_b U'^B)  - \mathcal{L}_{N'} (\bar{\mu}_q q^{ab}) N h_{AB} \ptl_a U'^A\ptl_b U^B 
\end{align}
likewise
\begin{align}
&  U'^A \mathcal{L}_{N'} \big( \ptl_b (N \bar{\mu}_q q^{ab} h_{AB} \ptl_a U^B) \big) \notag\\
&= U'^A \mathcal{L}_{N'} h_{AB}\ptl_a ( N \bar{\mu}_q q^{ab} \ptl_b U^B ) + U'^A \mathcal{L}_{N'} ( N \bar{\mu}_q q^{ab} \ptl_b U^B \ptl_{U^C} h_{AB} \ptl_a U^C ) .
\end{align}

Collecting the terms above, while using the background field equations \eqref{Kerr-p-dot} and computations analogous to the ones in Section 3; and
\begin{align}
& \ptl_a (U'^A \mathcal{L}_{N'} (N \bar{\mu}_q q^{ab} h_{AB} \ptl_a U^A)) \notag\\
& = \ptl_a U'^A  \mathcal{L}_{N'} (N \bar{\mu}_q q^{ab} h_{AB} \ptl_a U^A ) + \mathcal{L}_{N'} \big( \ptl_a (N \bar{\mu}_q q^{ab} h_{AB} \ptl_a U^A) \big) \notag \\
& = \ptl_a U'^A  \mathcal{L}_{N'} (N \bar{\mu}_q q^{ab} h_{AB} \ptl_a U^A ) + \ptl_a (\mathcal{L}_{N'} (N \bar{\mu}_q q^{ab} h_{AB} \ptl_a U^A))  
\end{align}
 we have,
\begin{align}
\ptl_t \mathbf{e}^{\text{Reg}} =& \ptl_b ( N^2 \bar{\mu}^{-1}_q ( \bar{\mu}_{q_0} q_{0}^{ab} p'_A \ptl_a U'^A) + U'^A \mathcal{L}_{N'} ( \bar{\mu}_{q_0} q^{ab}_0 h_{AB} \ptl_b U^B)  ) \notag\\
 & \mathcal{L}_{N'} (N) (2 \bar{\mu}_{q_0} q_{0}^{ab} \ptl_a \mbo{\nu}' + 2 \mathcal{L}_N \mbo{\nu}' \bar{\mu}_q q^{ab} \ptl_a N - 2 N'^b \bar{\mu}_q q^{bc} \ptl_a \mbo{\nu}' \ptl_c N) \notag \\
 &- H' \mathcal{L}_{N'}(N)
 \intertext{for $ \big \{ (q'_{ab}, \mbo{\pi}'^{ab}), (U'^A, p'_A) \big \} \in \mathscr{C}_{H'}$ this reduces to  }
 =& \ptl_b ( N^2 \bar{\mu}^{-1}_q ( \bar{\mu}_{q_0} q_{0}^{ab} p'_A \ptl_a U'^A) + U'^A \mathcal{L}_{N'} ( \bar{\mu}_{q_0} q^{ab}_0 h_{AB} \ptl_b U^B)  ) \notag\\
 & \mathcal{L}_{N'} (N) (2 \bar{\mu}_{q_0} q_{0}^{ab} \ptl_a \mbo{\nu}') + 2 \mathcal{L}_N \mbo{\nu}' \bar{\mu}_q q^{ab} \ptl_a N - 2 N'^b \bar{\mu}_q q^{bc} \ptl_a \mbo{\nu}' \ptl_c N). \notag \\
\end{align}
Thus, if we define, 
\begin{align}
(J^ t)^ {\text{Reg}} \fdg =& \, \mathbf{e}^{\text{Reg}} \notag\\
(J^b)^{\text{Reg}} \fdg=& \, N^2 e^{-2 \mbo{\nu}}(q_{0}^{ab} p'_A \ptl_a U'^A) + U'^A \mathcal{L}_{N'} ( \bar{\mu}_{q_0} q^{ab}_0 h_{AB} \ptl_b U^B)  ) \notag\\
 & \mathcal{L}_{N'} (N) (2 \bar{\mu}_{q_0} q_{0}^{ab} \ptl_a \mbo{\nu}') + 2 \mathcal{L}_N \mbo{\nu}' \bar{\mu}_q q^{ab} \ptl_a N - 2 N'^b \bar{\mu}_q q^{bc} \ptl_a \mbo{\nu}' \ptl_c N,
\end{align}
it follows that $J^{\text{Reg}}$ is a divergence-free vector field density for $\big \{ (q'_{ab}, \mbo{\pi}'^{ab}), (U'^A, p'_A) \big \} \in \mathscr{C}_{H'} \cap \mathscr{C}_{H'_a} \cap \mathscr{C}_{\tau'}.$ 

\end{proof}

\subsection*{Acknowledgements}
I acknowledge the gracious hospitality of Institut des Hautes \'Etudes Scientifiques (IHES) at Bures-sur-Yvette in Fall 2016 and the Department of Mathematics, Yale University during my postdoctoral stay,  where parts of this work were done. 
 Special gratitude is due to Vincent Moncrief for the encouragement. 
\bibliography{central-bib}
\bibliographystyle{plain}

\end{document}